\documentclass[aps,prx,reprint,superscriptaddress,nofootinbib]{revtex4-2}
\usepackage{graphicx}
\usepackage{xcolor}
\usepackage{caption}
\usepackage{subcaption}
\usepackage{amsmath}
\usepackage{amsthm}
\usepackage{xcolor}
\usepackage{amsfonts}
\usepackage{hyperref}
\usepackage{physics}
\usepackage{comment}
\usepackage{amssymb}
\usepackage{mathrsfs}
\usepackage{pifont}
\usepackage{bbold}
\usepackage{float}
\newtheorem{thm}{Theorem}

\newtheorem{lem}{Lemma}
\newtheorem{defn}{Definition}
\usepackage{multirow}
\usepackage[most]{tcolorbox} 
\newtheorem{protocol}{Protocol}
\def\be{\begin{equation}}
\def\ee{\end{equation}}
\def\bea{\begin{eqnarray}}
\def\eea{\end{eqnarray}}
\def\ben{\begin{equation*}}
\def\een{\end{equation*}}
\def\bean{\begin{eqnarray*}}
\def\eean{\end{eqnarray*}}
\def\bma{\begin{mathletters}}
\def\ema{\end{mathletters}}
\def\bi{\begin{itemize}}
\def\ei{\end{itemize}}
\def\bd{\begin{description}}
\def\ed{\end{description}}

\newcommand{\proj}[1]{\ket{#1}\bra{#1}}

\begin{document}

\title{Zero Knowledge Proofs in Quantum Networks}
\author{Tuhin Paul}
\email{tuhin.paul96@gmail.com}
\affiliation{Physics and Applied Mathematics Unit, Indian Statistical Institute, Kolkata 700108, India}
\author{Srijani Das}
\email{dassrijani33@gmail.com}
\email{srijani24_r@isical.ac.in}
\affiliation{Physics and Applied Mathematics Unit, Indian Statistical Institute, Kolkata 700108, India}
\author{Manasi Patra}
\email{manasipatra58@gmail.com}
\affiliation{Physics and Applied Mathematics Unit, Indian Statistical Institute, Kolkata 700108, India}

\author{Ramij Rahaman}
\email{ramijrahaman@isical.ac.in}
\affiliation{Physics and Applied Mathematics Unit, Indian Statistical Institute, Kolkata 700108, India}
\begin{abstract}

Zero-knowledge proofs (ZKPs) enable the verification of a statement without revealing any information beyond its validity and constitute a fundamental primitive in cryptography and information theory. However, existing constructions rely on computational assumptions and are predominantly confined to bipartite settings, leaving their information-theoretic realization in bipartite or network scenarios largely unexplored. Here we develop a framework for zero-knowledge verification based on the indistinguishability of quantum states under operational constraints. Exploiting the fundamental limitations imposed by local operations, we show that a verifier is inherently restricted from extracting information about the underlying state while retaining the ability to verify correctness. We construct explicit protocols for multiparty quantum networks that achieve information-theoretic security, ensuring that no subset of collaborating parties can gain knowledge beyond the validity of the statement, independent of their joint computational power. 
\end{abstract}
\maketitle

\section{Introduction}
The ability to verify the validity of a statement without revealing any additional information about the underlying proof constitutes one of the most profound and counterintuitive primitives in modern cryptography and information theory. Such protocols, known as zero-knowledge proofs (ZKPs), provide a rigorous framework to achieve this task, allowing a prover to convince a verifier of the validity of a statement while ensuring that no additional information about the proof itself is disclosed. Since their introduction by Goldwasser, Micali, and Rackoff \cite{Goldwasser1985Knowledge}, ZKPs have become indispensable tools for secure authentication, privacy-preserving verification, and blockchain protocols \cite{GMR89, Blckchn1,blckchn2,BenSasson2014,Bulletproofs2018,Narayanan2016}.\\ 
 Formally, a zero-knowledge proof protocol is typically formulated through an interactive exchange between two parties, a computationally unbounded prover and a polynomial-time verifier. Such protocols must satisfy three fundamental conditions: (i) completeness, which ensures that an honest prover can convince an honest verifier of a true statement with high probability; (ii) soundness, which guarantees that no dishonest prover can convince the verifier of a false statement except with negligible probability; and (iii) zero-knowledge, which requires that the verifier learns nothing beyond the validity of the statement.\\
 In the classical setting, the security of zero-knowledge protocols is inherently tied to computational assumptions. Most known constructions rely on the presumed hardness of certain mathematical problems, such as factorization \cite{RSA77} or discrete logarithms \cite{DH76}, and therefore provide only computational security. As a consequence, their security can be compromised by advances in algorithms or computational models. In particular, the advent of quantum computation challenges the foundations of classical cryptographic security, as quantum algorithms \cite{SHOR97,Grover97} can efficiently solve certain problems believed to be classically hard \cite{JOHNSON90,AroraBarak2009,GareyJohnson1979,Goldreich2008}.
 
 In 2006, Watrous~\cite{Watrous08} demonstrated that several canonical interactive proof systems, including the Goldreich-Micali-Wigderson protocols for graph isomorphism and graph 3-coloring, remain zero-knowledge against quantum verifiers. By constructing a quantum interactive proof system for a complete problem, he further established the equivalence of honest-verifier and general quantum statistical zero-knowledge, \(\mathrm{QSZK}_{\mathrm{HV}}=\mathrm{QSZK}\). Nevertheless, these protocols lack device independence, precluding information-theoretic security against uncharacterized devices. Moreover, while zero-knowledge protocols are well understood in the bipartite regime, their extension to genuine multipartite network architectures remains largely uncharted, particularly under information-theoretic constraints.
 
 These limitations motivate the development of information-theoretically secure protocols whose security is independent of any assumptions about the verifier’s computational power, and naturally lead to the study of zero-knowledge protocols in the quantum regime. In this work, we develop a zero-knowledge framework based on the indistinguishability of quantum states \cite{walgate2000,kar01,walgate02}, which fundamentally limits the information accessible to the verifier while retaining reliable verification of the claimed statement.We further construct explicit zero-knowledge protocols for multiparty quantum-network settings \cite{Kimble2008,Wehner2018}. Our constructions achieve information-theoretic security, ensuring that no subset of collaborating parties can extract any knowledge beyond the validity of the proven statement, independent of their computational power. This establishes, to our knowledge, the first rigorous framework for information-theoretically secure zero-knowledge proofs in multiparty quantum scenarios, thereby substantially extending the scope of zero-knowledge cryptography into the domain of distributed quantum information processing. 

A fundamental question is whether the identity of a shared entangled resource can be certified without being revealed. We address this question by introducing zero-knowledge certification of shared entanglement, considering Bell states \cite{Bell64,CHSH} and GHZ states \cite{GHZ} as paradigmatic resources. We first construct a zero-knowledge protocol for Bell-state certification and identify a correlation loophole that compromises its verification. We then introduce a two-basis certification procedure that closes this loophole, but show that it remains vulnerable to higher-dimensional realizations reproducing the target correlations without implementing the intended qubit-level structure. This observation reveals that secure certification must enforce both the incompatibility of the measurement observables and the effective two-dimensionality of the underlying systems, motivating a device-independent formulation. The device-independent (DI) framework is central to quantum-protocol security, enabling certification of shared correlations solely from observed input-output statistics, without assumptions about the internal workings or trustworthiness of the measurement devices. 
We extend the construction to multipartite quantum networks, where observed nonlocal correlations enable the self-testing of the shared GHZ state and the corresponding measurement observables, yielding a device-independent zero-knowledge protocol. We establish its zero-knowledge property information-theoretically by constructing a simulator whose induced verifier view is identical or, more generally, indistinguishable from that of the real protocol, without access to the prover's secret state. Finally, we show that these guarantees persist under noise, demonstrating the robustness of the protocols to imperfections in the shared quantum resources.

\section{Entanglement Certification based ZKP}
Entanglement constitutes a fundamental resource for a broad range of quantum information and quantum computation tasks. However, the mere presence of entanglement does not, by itself, guarantee the successful implementation of a given quantum protocol. Identifying and certifying the relevant structure of the shared entangled resource is equally crucial, since the operational utility of the resource depends not only on the presence of entanglement but also on the identity of the shared state. This is particularly evident in Bell tests, quantum teleportation, quantum key distribution, superdense coding, quantum random numbers generation, etc., where the performance and, in most of the cases, the security of the protocol depends critically on the specific Bell state shared by the parties.

Notably, knowledge of the identity of a shared Bell state by one party is sufficient to transform it into any of the four Bell states through an appropriate local unitary operation. This observation motivates a fundamental question: Can a prover, holding one subsystem of a shared Bell pair and knowing its identity, convince a verifier, holding the other subsystem, that the prover possesses this knowledge without revealing the identity of the shared state? This question naturally leads to the problem of zero-knowledge certification of shared entanglement, wherein the prover demonstrates knowledge of the shared entangled state while keeping its identity concealed. To formalize this problem, we begin with a simple two-party task, described by the following protocol.

\subsection{Bell-State Certification ZKP}\label{sec:p1}
\begin{protocol}{Single-Basis Bell-State Certification} \label{pro_1}

\noindent\textbf{Protocol Setting:} Let $P$ (Prover) and $V$ (Verifier) be two spatially separated parties who share \(N\) bipartite quantum systems arranged in an ordered sequence. Each shared pair is promised to be in one of the two Bell states. 
\begin{align*}
 \ket{\phi^{-}}=\frac{1}{\sqrt{2}}(\ket{00}-\ket{11})\\
 \ket{\psi^{+}}=\frac{1}{\sqrt{2}}(\ket{01}+\ket{10})
\end{align*}

The prover \(P\) possesses complete knowledge of the identity of the Bell state at each position \(i\in[N]\), whereas the verifier \(V\) knows only that each shared pair belongs to the set $\{\ket{\phi^-},\ket{\psi^+}\}$. Here and throughout, we use the notation $[N]:=\{1,2,\ldots,N\}$ for any natural number $N$.

\noindent\textbf{Goal:} The objective of the protocol is to enable \(P\) to convince \(V\) that \(P\) knows the identity of every shared Bell state, while revealing no information about the identities themselves beyond the fact that they belong to the prescribed set.\end{protocol}

\noindent\textbf{Protocol:}
A single round of the protocol is as follows:
\begin{enumerate}
\item \textbf{Challenge:} The verifier \(V\) selects an index \(i\in[N]\) uniformly at random and sends \(i\) to the prover \(P\).
\item \textbf{Prover's response:} Upon receiving \(i\), the prover \(P\) measures his subsystem of the \(i\)-th Bell pair in the Pauli-\(Z\) basis, obtaining an outcome \(m_P\in\{+1,-1\}\). Hereafter, we adopt the encoding $\pm 1 \mapsto \{0,1\}$ for measurement outcomes. Using this outcome together with his knowledge of the identity of the shared Bell state, \(P\) computes
 \[
 a_P=
 \begin{cases}
 m_P, & \text{for } \ket{\phi^-},\\[2mm]
 m_P\oplus1, & \text{for } \ket{\psi^+},
 \end{cases}
 \]
 and sends \(a_P\) to the verifier \(V\).
 
\item \textbf{Verification:} The verifier \(V\) measures his subsystem of the same Bell pair in the Pauli-\(Z\) basis, obtaining an outcome \(m_V\in\{0,1\}\). The verifier accepts if and only if \(
 a_P=m_V
 \).
\end{enumerate}

\noindent\textbf{ZKP Properties:} We assess the protocol according to the three standard properties of a zero-knowledge proof: completeness, soundness, and zero-knowledge.

\begin{itemize}
\item \textbf{Completeness:} For an honest prover who possesses the correct identity of the shared Bell state, the perfect correlations of the Bell pair ensure that the prover's response \(a_P\) agrees with the verifier's measurement outcome \(m_V\) with certainty. Consequently,
\begin{equation*}
 \Pr[\mathrm{accept}\mid P\ \mathrm{honest}]=1.
\end{equation*}
 
\item \textbf{Soundness:}
Suppose that a dishonest prover does not know whether the shared state is $\ket{\phi^-}$ or $\ket{\psi^+}$. From the prover's perspective, the reduced density matrix \[\rho_P=\rm{Tr}_V[\proj{\phi^-}=\dfrac{\mathbb{I}}{2}=\rm{Tr}_V[\proj{\psi^+}]\] are same and maximally mixed. 
Thus, without knowing the identity of the shared Bell state, the prover cannot predict the verifier's outcome better than random guessing. The maximum success probability in a single round is therefore $1/2$. 
For $r$ independent repetitions,
\begin{equation*}
\Pr(\text{accept})\leq 2^{-r}.
\end{equation*}
Hence, the soundness error decreases exponentially with the number of rounds.

\item \textbf{Zero-Knowledge}: The prover reveals only the predicted outcome $a_P$ and never discloses the underlying measurement outcome $m_P$. Moreover, for both $\ket{\phi^-}$ and $\ket{\psi^+}$, the verifier's reduced state is maximally mixed, \(\rho_V = \frac{I}{2}
\), and is therefore independent of the Bell-state identity. Hence, the verifier's local quantum state, together with the classical transcript received from the prover, carries no information about the identity of the shared Bell state beyond the prescribed prior knowledge. In particular, any strategy that enabled the verifier to distinguish $\ket{\phi^-}$ and $\ket{\psi^+}$ using only local operations and the protocol transcript would violate the no-signaling principle. Thus, the protocol is zero knowledge with respect to the hidden Bell-state identity. A rigorous simulator construction establishing the indistinguishability of the real and simulated transcripts is provided in section \ref{sec:simulator}. \end{itemize}

\noindent\textbf{Complexity:} Each protocol round requires only constant-time local computation by both parties, consisting of a single projective measurement followed by constant-time classical post-processing. The classical communication cost per round is $\lceil \log_2 N \rceil + 1$ bits, corresponding to the verifier's challenge index $i$ and the prover's one-bit response. Thus, the communication complexity of a single round is $O(\log_2 N)$. Repeating the protocol for $r$ rounds requires $O(r)$ local measurements steps and $O(r\log N)$ classical communication.\\

\noindent\textbf{Loophole in the Verification:} Although the protocol assumes that each shared pair is guaranteed to be in either $\ket{\phi^-}$ or $\ket{\psi^+}$, the $Z$-basis verification test does not certify this promise. Indeed, consider the separable state
\[
\rho_{PV}=\frac{1}{2}\left(\ket{00}\bra{00}+\ket{11}\bra{11}\right).
\]
It exhibits the same perfect $Z$-basis correlations required for acceptance: whenever $P$ obtains outcome $m_P$, the verifier obtains the same outcome. Thus, the prover can choose $a_P=m_P$ and pass the test certainty, despite possessing no information establishing the promised Bell-state structure. Hence, the verification test certifies only the observed $Z$-basis correlation and not the underlying state assumption, leaving a fundamental loophole in the soundness of the protocol.\\

\noindent\textbf{Closing the Correlation Loophole:}

For the separable state $\rho_{PV}$, the observed correlations are basis dependent. In particular, consider local measurements, for both $P$ and $V$, in a basis other than the computational basis,
\begin{align*}
\ket{a_0} &= \cos\theta\ket{0}+e^{i\gamma}\sin\theta\ket{1},\\
 \ket{a_1} &= \sin\theta\ket{0} - e^{i\gamma}\cos\theta\ket{1},
\end{align*}
 In general, $\rho_{PV}$ does not retain perfect correlations under such a basis change and is therefore distinguishable from the intended Bell-state correlations by a suitable choice of measurement setting.

This observation motivates a minimal modification of the verification procedure. Rather than fixing a single measurement basis, the verifier randomly selects between two incompatible observables, e.g., $Z$ and $X$. Although $\rho_{PV}$ reproduces the required perfect $Z$-basis correlation through classical correlations, it cannot simultaneously reproduce the corresponding perfect correlation in the complementary $X$ basis. Hence, testing both settings rules out the above separable-state strategy and strengthens the soundness of the protocol. More generally, a single-basis test certifies only the observed correlation, whereas complementary-basis tests constrain the underlying quantum state responsible for it.
 
\begin{protocol}{Two-Basis Bell-State Certification} \label{pro_2}

We now modify Protocol \ref{pro_1} to eliminate the single-basis correlation loophole described above. 

\noindent\textbf{Protocol setting}: The parties share $N$ bipartite systems, each promised to be in one of the two Bell states $\ket{\phi^-}$ or $\ket{\psi^+}$. The prover $P$ knows the identity of each shared state, whereas the verifier $V$ knows only the promised ensemble.

\noindent\textbf{Goal.} The prover convinces the verifier that he possesses the promised state information without revealing the identities of the shared Bell states.\end{protocol}

\noindent\textbf{Protocol.} In each round, the verifier and prover proceed as follows:
\begin{enumerate}
\item \textbf{Challenge.} The verifier samples $i\in[N]$ uniformly at random and independently chooses a measurement basis $b\in{Z,X}$ uniformly at random. He sends $(i,b)$ to $P$.

\item \textbf{Prover's response.} The prover measures his subsystem of the $i$-th bipartite state in basis $b$, obtaining $m_P\in\{0,1\}$. Using his knowledge of the Bell-state identity, he computes \begin{equation*}\label{ap_mp_2}
 a_P=\begin{cases}
 m_P&\text{for }\ket{\phi^-} \text{ and }b=Z\\
 m_P&\text{for }\ket{\psi^+} \text{ and }b=X\\
 m_P\oplus 1&\text{for }\ket{\phi^-} \text{ and }b=X\\
 m_P\oplus 1&\text{for }\ket{\psi^+} \text{ and }b=Z \end{cases}
\end{equation*}

and sends \(a_P\) to the verifier \(V\).

\item \textbf{Verification.} The verifier measures his subsystem of the same state in the basis $b$, obtaining $m_V\in\{0,1\}$, and accepts if and only if
\(a_P=m_V\). 
\end{enumerate}

\noindent\textbf{ZKP Properties}: For either choice \(b\in{Z,X}\), the two possible states require opposite correlation predictions: \(\ket{\phi^-}\) exhibits correlation in the \(Z\) basis and anticorrelation in the \(X\) basis, whereas \(\ket{\psi^+}\) exhibits the complementary pattern. Moreover, the reduced states on the prover's subsystem are identical, \(\rho_P=\frac{I}{2}\) for both the state, so no local measurement performed by a dishonest prover, who does not know the state identity, can reveal the identity of the shared state. Consequently, in the absence of additional information, the optimal strategy is to guess the required correlation, yielding
\[
\Pr(\text{accept})\leq \frac{1}{2}.
\]
Thus, the completeness, soundness, and zero-knowledge arguments remain unchanged from Protocol \ref{pro_1}.\\

\textbf{Complexity:} Each round involves one local projective measurement and constant-time classical post-processing. The communication cost is $\lceil\log_2 N\rceil+2=O(\log N)$ bits per round, accounting for the index, basis choice, and one-bit response. Thus, over $r$ rounds, the protocol requires $O(r)$ local operations and $O(r\log N)$ classical communication, with $r$ shared Bell-state uses. This complexity is asymptotically identical to that of Protocol \ref{pro_1}, differing only in constant communication overhead.\\
The two-basis test also eliminates the single-basis loophole of Protocol \ref{pro_1}. In particular, the separable state
\[
\rho_{PV}=\frac{1}{2}\left(\ket{00}\bra{00}+\ket{11}\bra{11}\right)
\]
reproduces the perfect correlations required in the \(Z\) basis but cannot simultaneously reproduce the corresponding anticorrelations in the incompatible \(X\) basis. Hence, successful verification cannot be achieved by reproducing the statistics of a single measurement setting. The use of complementary bases therefore rules out this separable-state strategy and closes the single-basis loophole.

It is important to note that the Protocols \ref{pro_1} and \ref{pro_2} can be straightforwardly generalized to the complete Bell-state ensemble,

\begin{align*}
\ket{\phi^\pm}&=\frac{1}{\sqrt{2}}\left(\ket{00}\pm\ket{11}\right),\\
\ket{\psi^\pm}&=\frac{1}{\sqrt{2}}\left(\ket{01}\pm\ket{10}\right).
\end{align*}
The resulting protocols retain the same essential structure and inherit analogous completeness, soundness, and zero-knowledge properties.\\

\noindent\textbf{LOCC state distinguishability \& ZKP}: State distinguishability and zero-knowledge proofs (ZKP) constitute fundamentally different tasks, although they may be operationally related in some cases. In state distinguishability, two parties share an unknown state drawn from a known ensemble and seek to identify it under restricted operations, such as local operations and classical communication (LOCC). By contrast, in a ZKP, the prover $P$ seeks to convince the verifier $V$ that $P$ knows the identity of the shared state, without revealing any information about that identity.

For example, the four Bell states\(\{\ket{\phi^\pm},\ket{\psi^\pm}\}\) cannot be perfectly distinguished by LOCC \cite{kar01}. Nevertheless, when one of these states is shared between \(P\) and \(V\), Protocol \ref{pro_2} enables P to demonstrate knowledge of its identity to V without revealing the identity itself.\\

\noindent\textbf{Loophole in Protocol \ref{pro_2}}: At first sight, the protocol appears loophole-free, provided that the measurement devices implement genuinely incompatible measurements and that each subsystem is guaranteed to be a qubit. However, in the absence of independent certification of these assumptions, the protocol admits a higher-dimensional separable-state simulation. For example,
\[
\rho_{PV}
=
\frac12
\left(
|\Psi_1\rangle\!\langle\Psi_1|
+
|\Psi_2\rangle\!\langle\Psi_2|
\right),
\]
where
\begin{align*}
|\Psi_1\rangle
&=
|+\rangle_Z^{P_1}|+\rangle_X^{P_2}
|+\rangle_Z^{V_1}|-\rangle_X^{V_2} \text{ and}\\
|\Psi_2\rangle
&=
|-\rangle_Z^{P_1}|-\rangle_X^{P_2}
|-\rangle_Z^{V_1}|+\rangle_X^{V_2}, 
\end{align*}
is a classical mixture of orthogonal product states.\\ Defining the local measurement observables for \(P(V)\) as
\[
Z_{P(V)}=\sigma_z^{P_1(V_1)}\otimes I^{P_2(V_2)}\text{ and }
X_{P(V)}=I^{P_1(V_1)}\otimes\sigma_x^{P_2(V_2)}.
\]
one obtains \(
\langle Z_PZ_V\rangle=1\) and \(
\langle X_PX_V\rangle=-1\),
which exactly reproduce the correlations of the Bell state
\(\ket{\phi^-}\). Thus, by exploiting the indistinguishability of the measurement statistics produced by this separable, higher-dimensional state from those of the target Bell state \(\ket{\phi^-}\), the prover can successfully deceive the verifier. This demonstrates that protocol security requires not only certification of the incompatibility of the measurement observables but also verification of the effective two-dimensionality of the underlying quantum systems. These considerations motivate the device-independent protocol introduced in the following section. 

\begin{protocol}{Device-independent ZKP: Bell-State Certification} \label{pro_3}

\noindent\textbf{Protocol Setting:} 
The prover $P$ and verifier $V$ share $N$ maximally entangled two-qubit states, each of which is promised to be one of the four Bell states, $\{|\phi^\pm\rangle,|\psi^\pm\rangle\}$. The prover knows the identity of each shared state, whereas the verifier has no information about the individual state identities.

\noindent\textbf{Goal:}
The objective is for $P$ to convince $V$ that the selected shared state is a maximally entangled Bell state, while revealing no information about which Bell state it is.
\end{protocol}

\noindent\textbf{Protocol:}
\begin{enumerate}
\item \textbf{Challenge:}
The verifier randomly selects a subset of the shared states and communicates the corresponding choices to the prover.
\item \textbf{State Transformation:}
For each selected state, the prover exploits his knowledge of its identity to apply the corresponding local unitary transformation
\[
U_P=
\begin{cases}
I, \text{if the state is }\ket{\psi^-},\\
\sigma_X, \text{if the state is }\ket{\phi^-},\\
\sigma_Y, \text{if the state is }\ket{\phi^+},\\
\sigma_Z, \text{if the state is }\ket{\psi^+}.
\end{cases}
\]
Thus, irrespective of the initial Bell-state identity, the selected state is mapped to \(\ket{\psi^-}\), up to an irrelevant global phase. This transformation removes the dependence of the subsequent verification procedure on the original state identity. The prover does not reveal any information about the applied transformation \(U_P\) to the verifier.

\item \textbf{CHSH Test:}
 For each selected state, the prover and verifier perform local measurements according to the settings
\[
P_1=\frac{-Z-X}{\sqrt2},\qquad
P_2=\frac{Z-X}{\sqrt2},
\]
for \(P\), and
\[
V_1=X,\qquad V_2=Z
\] for \(V\). Both parties record their measurement settings and corresponding outcomes in each run.
\item \textbf{Verification:}
 For each run, \(P\) sends his measurement setting and outcome to \(V\). The verifier then evaluates the CHSH correlator \[
S=
\langle P_1V_1\rangle+
\langle P_1V_2\rangle+
\langle P_2V_1\rangle-
\langle P_2V_2\rangle
\] from his local measurement data and the information received from \(P\). The protocol is accepted if the observed value is consistent with the maximal quantum violation, namely,
\[
S \ge 2\sqrt2-\epsilon,
\]
where \(\epsilon>0\) is a sufficiently small tolerance accounting for finite-statistics and experimental imperfections.

\end{enumerate}

\noindent\textbf{ZKP Properties:}

\begin{itemize}

\item \textbf{Completeness:} For an honest prover, who knows the identity of each selected state, the prescribed local unitary maps every selected state to $\ket{\psi^-}$ up to an irrelevant global phase. The subsequent CHSH measurements therefore yield
\(S\simeq2\sqrt2\) corresponding to the maximal quantum violation and provide a device-independent self test for $\ket{\psi^-}$ \cite{Yao}. Thus, in the ideal limit, an honest prover is accepted with unit probability, while finite-statistics and experimental imperfections result only in negligible deviations from perfect completeness. The robustness of the scheme against noise is discussed in section \ref{sec:robust}. 
\item \textbf{Soundness:} A prover who does not possess the promised Bell-state information is effectively described by the maximally mixed two-qubit separable state
\[
\rho_{PV}
=\frac{1}{4}\sum_{\beta\in\{\phi^\pm,\psi^\pm\}}
\ket{\beta}\bra{\beta}
=\frac{\mathbb{I}_2^P}{2}\otimes \frac{\mathbb{I}_2^V}{2}.
\]
Being separable, this state cannot violate the CHSH inequality and hence satisfies \(S\leq 2\), the local-realistic bound \cite{Bell64, CHSH}. Consequently, a prover lacking the promised Bell-state information cannot reproduce the correlations required by the prescribed Bell test and, hence, cannot convince the verifier.

We next consider a prover possessing only partial information about the shared Bell state. With probability \(p\), the prover correctly identifies the Bell state, while with probability \(1-p\),he has no information about its identity. The resulting strategy is thus a convex mixture of the informed and uninformed strategies. By convexity of the CHSH expression, the effective CHSH value satisfies
\begin{align*}
 S_{\rm eff}&\leq p S_{\rm inf}+(1-p)S_{\rm uninformed}\\
 &\leq 2(1+p(\sqrt{2}-1)),
\end{align*}

where \(S_{\rm inf}\) and \(S_{\rm uninformed}\) denote the maximal CHSH values attainable with complete and no Bell-state information, respectively. Hence, the deficit from the maximal quantum value $2\sqrt{2}$ is $2\sqrt{2}-S_{\rm eff}\ge 2(\sqrt{2}-1)(1-p)\simeq 0.0828 (1-p)$. Thus, unless $p$ is sufficiently close to unity, the resulting CHSH violation remains appreciably below the Tsirelson bound, preventing the prover from convincing the verifier within a verification test requiring near-maximal CHSH violation.

\item \textbf{Zero Knowledge:}
The prover's state-dependent operation is local and is not revealed to the verifier. Moreover, for every Bell state, the marginal state of the verifier is, 
\[
\rho_V
=
\operatorname{Tr}_P
\!\left(
|\beta\rangle\langle\beta|
\right)
=
\frac{\mathbb{I}}{2},\]
where, \(
|\beta\rangle\in
\{
|\phi^\pm\rangle,
|\psi^\pm\rangle
\}\). Since the prover's local transformation leaves $\rho_V$ invariant, the verifier's reduced state and hence his local statistics are independent of the original Bell-state identity. The CHSH transcript therefore reveals no information about which Bell state was initially shared.\end{itemize}
\textbf{Complexity:} The complexity of the protocol is determined by the number of states $r$ selected for CHSH verification and the number of measurements performed on each selected state. For each selected state, the prover performs one local unitary followed by constant-time classical post-processing. The prover communicates one measurement setting and one-bit outcome per CHSH run, resulting in $O(1)$ classical communication per run. For $r$ independent CHSH runs, the total local computational and measurement costs are $O(r)$, while the prover-to-verifier communication is $O(r)$ bits. Specifying the $r$ randomly selected states requires $O(r\log N)$ bits. Hence, the total communication complexity is $O(r\log N)$, while the verification complexity is $O(r)$. To obtain a CHSH violation satisfying $S\geq 2\sqrt{2}-\epsilon$, the protocol requires $r=O(1/\epsilon^2)$ samples, yielding $O(1/\epsilon^2)$ local measurements and computational cost, $O(1/\epsilon^2)$ prover-to-verifier communication, and $O(\log N/\epsilon^2)$ total classical communication.

We now extend the entangled-state certification ZKP to a multipartite network. Specifically, a prover $P_1$ seeks to convince $n-1$ verifiers, $V_2,V_3,\dots,V_{n}$ of his knowledge of a randomly selected sequence of shared $n$-qubit GHZ states, $\displaystyle\ket{\phi_n^\pm}=\dfrac{1}{\sqrt{2}}\left(\ket{0}^{\otimes n}\pm \ket{1}^{\otimes n}\right)$. 

\subsection{GHZ certification ZKP}
\begin{protocol}{ZKP: Multipartite State Certification} \label{pro_m}

\noindent\textbf{Protocol Setting:}
Consider $n$ spatially separated parties, comprising a prover $P_1$ and $n-1$ verifiers $V_2,V_3,\dots, V_{n}$, who share $N$ copies of an $n$-qubit state, with each copy promised to be one of the two GHZ states
\[
\ket{\phi_n^{\pm}}
=\frac{1}{\sqrt2}(\ket{0}^{\otimes n}\pm\ket{1}^{\otimes n}).
\]
The prover has complete knowledge of the identity of each shared state, whereas the verifiers know only the promised ensemble.

\noindent\textbf{Goal:}
The objective is for $P_1$ to convince the verifiers $V_2,V_3,\dots, V_{n}$ that he knows the identity of each shared state, without revealing any information about it.
\end{protocol}
\noindent\textbf{Protocol:}

\begin{enumerate}
\item \textbf{Challenge:} In each round, the verifiers randomly select \(V_k\) and \(r\in[N]\), uniformly and independently. Each verifier \(V_j\) measures their subsystem of the \(r\)-th shared state in the \(X\) basis, obtaining \(m_{V_j}\). All verifiers \(V_j\neq V_k\) broadcast \((j,m_{V_j})\), while \(V_k\) withhold his outcome \(m_{V_k}\).

\item \textbf{Prover Response:}

The prover measures their subsystem of the \(r\)-th shared state in the \(X\) basis, obtaining \(m_P\). Knowing the identity of the shared GHZ state and the announced outcomes \({a_{V_j}}\) for \({j\neq k}\), the prover predicts the withhold outcome of \(V_k\) as
\[
a_{V_k}=
\begin{cases}
m_P\bigoplus_{j\neq k}m_{V_j},
& \text{if the shared state is }\ket{\phi_n^+},\\
1\bigoplus m_P\bigoplus_{j\neq k}m_{V_j},
& \text{if the shared state is }\ket{\phi_n^-}.
\end{cases}
\]

The prover sends \(a_{V_k}\) to \(V_k\).

\item \textbf{Verification:}
The verifier \(V_k\) accepts the round iff
\[
m_{V_k}=a_{V_k}.
\]
The test is repeated for a sufficiently large number of rounds, with \(V_k\) selected uniformly at random in each round, ensuring that every verifier serves as the challenge verifier in a non-negligible fraction of the rounds.
\end{enumerate}

\noindent\textbf{ZKP Properties:}
\begin{itemize}
\item \textbf{Completeness:} 
For the states \(\ket{\phi_n^{+}}\) and \(\ket{\phi_n^{-}}\), the parity of the outcomes of local \(X\)-basis measurements is, respectively, even and odd. Hence, if the prover knows the identity of the shared GHZ state and the \(X\)-basis outcomes of all verifiers \(V_j\neq V_k\), the outcome of the remaining verifier \(V_k\) is uniquely determined, allowing the prover to predict \(a_{V_k}\) as \(m_{V_k}=a_{V_k}\) with certainty. Thus, \[
\Pr(\mathrm{accept})=1.
\]

\item \textbf{Soundness:} 
If the prover lacks knowledge of the shared GHZ-state identity, the corresponding correlation pattern is inaccessible to him, and his probability of correctly predicting the outcome of \(V_k\) is bounded by
\[
\Pr(\text{accept})\leq \frac12.
\]
Consequently, after \(r\) independent rounds, the soundness error is exponentially suppressed as
\[
\Pr(\text{accept})\leq 2^{-r}.
\]

\item \textbf{Zero-Knowledge:} 
The reduced states of every proper subsystem \(S\subsetneq[n]\) are identical for \(\ket{\phi_n^{+}}\) and \(\ket{\phi_n^{-}}\) \[\Tr_{S}\left[\proj{\phi_n^+}\right]=\Tr_{S}\left[\proj{\phi_n^-}\right].\] 
Consequently, even under arbitrary joint measurements on their systems, the verifiers obtain no information about the identity of the shared state, provided they have no access to the prover's quantum system. This indistinguishability follows directly from the no-signaling principle, and hence the protocol is information-theoretically zero-knowledge.
\end{itemize}

\textbf{Complexity}: Each round requires a single local projective measurement in the Pauli-\(X\) basis by each of the \(n\) parties, together with constant-time classical post-processing by the prover. The classical communication per round comprises the broadcast of the selected index \(i\in[N]\), requiring \(\lceil\log_2 N\rceil\) bits, the $(n-2)$ one-bit measurement outcomes from the verifiers \(V_j\neq V_k\), and the prover's one-bit prediction. Thus, the communication cost per round is \(O(\log N+n)\).
After \(r\) independent rounds, the protocol has computational complexity \(O(r)\), communication complexity
\[
O\left(r(\log N+n)\right),
\]
and consumes \(r\) shared GHZ states.

\noindent\textbf{Verification Loophole}: Although the protocol assumes that each shared state is promised to be either $\ket{\phi_n^+}$ or $\ket{\phi_n^-}$, this promise is not certified by the verification test itself. In particular, the separable state
\[
\rho_{P_1V_2V_3\dots V_{n}}
=
\frac{1}{2}
\left(
\ket{+}\bra{+}^{\otimes n}
+
\ket{-}\bra{-}^{\otimes n}
\right)
\]
reproduces the perfect \(X\)-basis correlations required by the protocol, with all parties obtaining identical outcomes in every round. Consequently, a dishonest prover can exploit this separable state to pass the verification test without possessing the promised GHZ entanglement. This reveals a fundamental limitation of the prepare-and-measure verification and motivates a self-testing formulation, in which the observed nonlocal correlations certify the underlying GHZ state and measurement observables, thereby establishing a device-independent protocol as described below.

\begin{protocol}{Device-independent ZKP: GHZ-State Certification} \label{pro_m2}

\noindent\textbf{Protocol Setting:}
Consider \(n\) (odd) spatially separated parties, a prover \(P_1\) and \(n-1\) verifiers \(V_2,V_3,\ldots,V_{n}\), sharing \(N\) states each of which is one of four GHZ states,
\begin{align*}
 \ket{\phi_n^{\pm}}
&=\frac{1}{\sqrt2}(\ket{0}^{\otimes n}\pm\ket{1}^{\otimes n})\\
\ket{\psi_n^{\pm}}
&=\frac{1}{\sqrt2}\big(\ket{0}\ket{1}^{\otimes (n-1)}\pm\ket{1}\ket{0}^{\otimes (n-1)}\big).\end{align*}
For each shared copy, the prover \(P_1\) holds the first qubit, while the remaining \(n-1\) qubits are distributed among the \(n-1\) verifiers. The prover knows the identity of each shared state, whereas the verifiers know only the underlying ensemble. 

\noindent\textbf{Goal:}
The goal is to certify the prover's knowledge of the state identities while revealing no information beyond the validity of the claim.
\end{protocol}

\noindent\textbf{Protocol:}

\begin{enumerate}
\item \textbf{Challenge:} In each round, the verifiers uniformly select a state index \(\mu\in[N]\) and an operator index \(i\in\{0,1,\ldots,n\}\), corresponding to \be\label{operator}\begin{aligned}
 \hat{\mathcal O}_0&=X_1X_2\cdots X_n,\\\hat{\mathcal O}_i
&=
X_1\cdots X_{i-1}Y_iY_{i+1}X_{i+2}\cdots X_n,
\quad \forall
 i\in[n],
\end{aligned} \ee
with \(n+1\equiv1\pmod n\). Here, \(X_j=\sigma_x\) and \(Y_j=\sigma_y\) denote the local measurement observables of the \(j\)-th party. The verifiers then communicate the selected pair \((\mu,i)\) to the prover. 
 
\item \textbf{Prover Operation:} Using his private knowledge of the state identity, the prover applies the corresponding local unitary

\[
U_P=
\begin{cases}
I, \text{if the state is }\ket{\phi_n^-},\\
\sigma_Z, \text{if the state is }\ket{\phi_n^+},\\
\sigma_X, \text{if the state is }\ket{\psi_n^-},\\
\sigma_Y, \text{if the state is }\ket{\psi_n^+}
\end{cases}
\] to his subsystem of the selected state, thereby mapping it to \(\ket{\phi_n^-}\) up to an irrelevant global phase \(e^{i\theta_\chi}\): 
\[
(U_P\otimes I^{\otimes n-1})\ket{\chi}
=e^{i\theta_\chi}\ket{\phi_n^-},
\ket{\chi}\in
\left\{\ket{\phi_n^\pm},\ket{\psi_n^\pm}\right\}.
\]
The prover subsequently measures his subsystem of the selected state according to the operation label \(i\) communicated in the preceding step.
\item \textbf{State certification:} Following the prover’s measurement and announcement of his outcome, the verifiers perform the prescribed measurements on their respective systems and announce their outcomes. They then verify whether the resulting outcomes satisfy the eigenvalue relations
\begin{equation}\label{eigen_value}\begin{aligned} 
\hat{\mathcal O}_0 \ket{\phi_n^-}&=(-1)\,\ket{\phi_n^-}\\\hat{\mathcal O}_i \ket{\phi_n^-}&=(+1)\,\ket{\phi_n^-}, \quad \forall
 i\in[n].
 \end{aligned}\end{equation}
\item \textbf{Verification:} By Theorem \ref{th}, satisfaction of the eigenvalue relations in Eq. \eqref{eigen_value} self-tests the selected correlations to the state $\ket{\phi_n^-}$ for odd $n$, thereby certifying that the prover possesses knowledge of the identity of shared ensemble.
\end{enumerate}

Protocol \ref{pro_m2} can be implemented using only the two states \(\ket{\phi_n^\pm}\), with the corresponding completeness, soundness, and zero-knowledge analysis remaining unchanged. Conversely, Protocol \ref{pro_m} extends directly to the four-state ensemble employed in Protocol \ref{pro_m2}, provided that the first qubit is always held by the prover. In this case, the additional state sector can be interconverted locally by the prover through a \(\sigma_X\) bit flip on the first qubit, after which the original Protocol \ref{pro_m} procedure applies without modification. Since this local preprocessing is determined solely by the prover's private state knowledge and is not revealed to the verifiers, the completeness, soundness, and zero-knowledge properties of Protocol \ref{pro_m} remain unchanged.\\

\noindent\textbf{ZKP Properties:}
\begin{itemize}
 \item \textbf{Completeness:} If the prover knows the identity of every challenged state and applies the corresponding unitary \(I,\sigma_Z,\sigma_X,\sigma_Y\), each selected state is transformed into \(\ket{\phi_n^-}\). Therefore, for any challenge $(\mu,i)$ chosen by the verifiers, the eigenvalue relations \eqref{eigen_value}
in Step 3. is satisfied with certainty and maximum Bell value
\(
\big\langle {\mathcal B}_n\big\rangle
= n+1
\) defined in Eq. \eqref{max_bell} attained deterministically. Thus the protocol has \emph{perfect completeness}:
\[
\Pr[\text{accept}\mid P \text{ honest}] = 1.
\]
 \item \textbf{Soundness:} To evaluate soundness, consider first a dishonest prover with no knowledge of the target state’s identity. The effective state reduces to the maximally mixed ensemble over the four-dimensional GHZ basis,\begin{equation}\label{eq_mix}\rho_0 = \frac{1}{4} \sum_{\mu \in { \pm }} \left( |\phi_n^\mu \rangle \langle \phi_n^\mu| + |\psi_n^\mu \rangle \langle \psi_n^\mu| \right).\end{equation}Because $\mathrm{Tr}(\rho_0 \hat{\mathcal{O}}_i) = 0$ for all $i \in \{0, \dots, n\}$, any local-realistic (LR) strategy can satisfy at most $n$ of the $n+1$ eigenvalue constraints in Eq.~\eqref{eigen_value}. Over $k$ independent verification rounds, the acceptance probability is bounded by\begin{equation*}P_{\mathrm{accept}} \le \left( \frac{n}{n+1} \right)^k = \exp\left[ -k \ln\left( 1 + \frac{1}{n} \right) \right],\end{equation*}suppressing the soundness error exponentially in $k/n$. When the prover possesses partial information-identifying the target state with prior probability $p \in (0, 1)$ a desired local unitary transformation prepares the effective shared state\begin{equation*}\sigma = p |\phi_n^-\rangle \langle \phi_n^-| + (1-p)\rho_0,\end{equation*} with $\rho_0$ defined as in Eq.~\eqref{eq_mix}. In this case, the single-round acceptance probability satisfies $P_{\mathrm{succ}}^{(1)} \le p + (1-p)n/(n+1) = 1 - (1-p)/(n+1)$. Consequently, across $k$ independent rounds, the soundness error obeys\begin{equation*}P_{\mathrm{accept}} \le \exp\left[ -k \ln\left( 1 + \frac{1-p}{n+p} \right) \right].\end{equation*}

 \item \textbf{Zero knowledge:} 
The prover's correction operation
$U_P \in \{I, \sigma_Z, \sigma_X, \sigma_Y\}$ maps \emph{every} one of the four
possible states onto the
\emph{same} canonical state $\ket{\phi_n^{-}}$ \emph{before} any measurement is performed, so that the secret state identity is
effectively erased from the system prior to the generation of any statistics
visible to the verifiers. Since for all the parties $\ket{\phi_n^{-}}$ has maximally mixed marginal states, every individual measurement outcome is perfectly
random.\\
Even if the $n-1$ verifiers are permitted to bring their individual qubits
together and perform a \emph{joint} (global) measurement, no additional
information about the state identity is gained. Under such a collective
measurement, the reduced state held by the verifiers is, for each of the four
possible GHZ states $\{\ket{\psi_n^{\pm}},\ket{\phi_n^{\pm}}\}$, formally
equivalent to one of the four maximally entangled two-qubit Bell states
$\{\ket{\Phi^{\pm}},\ket{\Psi^{\pm}}\}$, with the collective $(n-1)$-qubit block
playing the role of a single effective qubit. Since the four Bell states are
locally indistinguishable to any party lacking access to the complementary
subsystem held by the prover~\cite{kar01}, the verifiers, even acting in
concert, cannot determine which of the four states was shared. The protocol
therefore remains zero-knowledge under collective verifier measurements, as no
strategy, local or global, on the verifiers' side can extract the state
identity beyond the single bit certifying the validity of the prover's claim. Moreover, because this
certification relies only on the observed correlations reaching the algebraic
maximum $n+1$ via the self-testing/rigidity argument rather than on trusted
device descriptions, the same no-extra-leakage guarantee extends even to
verifiers with untrusted measurement apparatus, making the scheme
\emph{device-independently zero-knowledge}.\end{itemize}

\noindent\textbf{Complexity.} Let \(r\) denote the number of randomly sampled copies subjected to verification. In each round, the challenge \((\mu,i)\) requires
\(\lceil\log_2 N\rceil+\lceil\log_2(n+1)\rceil\) bits, while the prover performs one single-qubit Pauli operation and one local measurement, and each of the \(n-1\) verifiers performs one local \(X/Y\) measurement and announces one binary outcome. Hence, the total local measurement/unitary-operation cost is \(O(rn)\), and the classical communication complexity is
\(
O\!\left[r\bigl(\log N+\log n+n\bigr)\right]
=O\!\left[r(n+\log N)\right]\ \text{bits}.
\)
The verification procedure consists of checking the eigenvalue/parity relations for all \(r\) rounds. Hence the overall computational complexity is \(O(rn)\). If \(r=O(\epsilon^{-2}\log(1/\delta))\) samples are required to estimate the relevant correlations to additive accuracy \(\epsilon\) with failure probability at most \(\delta\), the computational complexity $C_{\mathrm{comp}}
$ and the communication complexity $C_{\mathrm{comm}}$ becomes,
\ben\begin{aligned}
C_{\mathrm{comp}}
&=O\!\left(n\epsilon^{-2}\log\frac{1}{\delta}\right),\\
C_{\mathrm{comm}}&
=O\!\left[(n+\log N)\epsilon^{-2}\log\frac{1}{\delta}\right]. 
\end{aligned}
\een For constant \(\epsilon\) and \(\delta\), these reduce to \(O(nr)\) computational complexity and \(O[r(n+\log N)]\) classical communication.
\section{Simulation of the Verifier's View} 
\label{sec:simulator}
The zero-knowledge property is established by showing that the verifier's view can be simulated without access to the prover's secret state. Specifically, we compare the superoperator describing the verifier's view in the real protocol with a simulator superoperator constructed independently of the prover's secret information. Equality, or more generally indistinguishability, of these superoperators ensures that the verifier gains no information about the witness beyond that implied by the validity of the statement. In what follows, we explicitly construct the simulators for the two device-independent protocols, (\ref{pro_3}) and (\ref{pro_m2}). The zero-knowledge property of the remaining protocols follows directly from these constructions. We consider malicious quantum-capable verifiers whose auxiliary system is initially independent of the prover-verifier Bell state. For arbitrary quantum verifiers, we initially assume that the auxiliary quantum state is tensor-product with the prover–verifier entangled resource, corresponding to the idealized setting. In general, however, such a factorization need not hold in the presence of imperfections or coherent correlations. These coherent attacks and the corresponding robust treatment are addressed in Section \ref{sec:robust}.

\subsection{Simulator for Protocol \ref{pro_3}}
\noindent\textbf{Transcript Super-operator:} Let each of the shared state $\ket{\beta_{s_j}}$ be prepared in one of the bell states, $\lvert\beta_{s_j}\rangle
\in
\{
\lvert\phi^\pm\rangle,
\lvert\psi^\pm\rangle
\}$ and then the shared state be represented as,\[
|\Phi_{\mathcal{S}}\rangle_{PV}
=
\bigotimes_{j=1}^{N}
|\beta_{s_j}\rangle_{P^jV^j}.\]
The identity of each of the states i.e., the string 
\(
\mathcal S=(s_1,\ldots,s_N)
\)
is known only to the prover.\\ In each protocol round, the verifier selects uniformly at random a subset
\(
\Lambda\subseteq[N],~~
|\Lambda|=k.
\) The set of all $k$-element subsets is denoted by
\(
 \tilde{\Lambda}_k
=
\left\{
\Lambda\subseteq[N]:|\Lambda|=k
\right\},
\)
\(
|\tilde{\Lambda}_k|
=
\binom{N}{k}.
\)

For each Bell state, the prover applies a local unitary
$U_{s_j}$ such that
\begin{equation*}
(U_{s_j}\otimes \mathbb{I})
\ket{\beta_{s_j}}
=
e^{i\theta_{s_j}}
\ket{\psi^-}.
\end{equation*}
For a selected subset $\Lambda$, defining the collective correction as, $U_\Lambda(\mathcal S)
=
\bigotimes_{i\in \Lambda}U_{s_i}$, we have 
\begin{align*}
&
\left(
U_\Lambda(\mathcal S)\otimes \mathbb{I}_{V}^{\otimes\Lambda}
\right)
\ket{\Phi_{\mathcal S}}
\nonumber\\
&\quad =
e^{i\Theta_{\mathcal S,\Lambda}}
\left(
\bigotimes_{i\in \Lambda}
\ket{\psi^-}_{P^iV^i}
\right)
\otimes
\left(
\bigotimes_{j\notin \Lambda}
\ket{\beta_{s_j}}_{P^jV^j}
\right),
\label{eq:subset_corrected_state}
\end{align*}
where \(\Theta_{\mathcal S,\Lambda}
=
\sum_{i\in \Lambda}\theta_{s_i}.\) Since the phase is global, it has no physical
consequence. Therefore, the selected subsystem is effectively
\(
\ket{\psi^-}^{\otimes k},
\)
independent of the secret string $\mathcal S$. Tracing out all prover systems for all the selected states as well as the unselected bell states gives us, \begin{align*}
\rho_V^{\mathcal S}
&=
\operatorname{Tr}_{P}
\left[
\left(
U_\Lambda(\mathcal S)\otimes \mathbb{I}
\right)
\ket{\Phi_{\mathcal S}}
\bra{\Phi_{\mathcal S}}
\left(
U_\Lambda^\dagger(\mathcal S)\otimes \mathbb{I}
\right)
\right]=
\bigotimes_{j=1}^{N}\frac{\mathbb{I}_{V^j}}{2}.
\end{align*}

For every selected state $i\in \Lambda$, the verifier chooses CHSH
measurement settings
\(
x_i,y_i\in\{1,2\},
\)
independently and uniformly.The prover's measurement settings are,
\begin{equation*}
\label{eq:prover_observables}
P_1
=
-\frac{Z+X}{\sqrt{2}},
\qquad
P_2
=
\frac{Z-X}{\sqrt{2}},
\end{equation*}
while the verifier uses
\(
V_1=X,
\qquad
V_2=Z\).

Let
\(
a_i,b_i\in\{0,1\}
\)
denote the measurement outcomes of the prover and verifier, respectively. The corresponding projectors are
\begin{equation*}
\label{eq:projectors}
\Pi_{a_i}^{P_{x_i}}
=
\frac{1}{2}
\left[
\mathbb{I}+(-1)^{a_i}P_{x_i}
\right], \quad \Pi_{b_i}^{V_{y_i}}
=
\frac{1}{2}
\left[
\mathbb{I}+(-1)^{b_i}V_{y_i}
\right]
\end{equation*}For the selected subset \(\Lambda\), define the collective projectors for the prover and verifier, respectively, as
\begin{equation*}
M_{\mathbf a}^{\Lambda,\mathbf x}
=
\bigotimes_{i\in \Lambda}
\Pi_{a_i}^{P_{x_i}},~~ \text{and } ~ N_{\mathbf b}^{\Lambda,\mathbf y}=\bigotimes_{i\in \Lambda}
\Pi_{b_i}^{V_{y_i}},
\label{eq:collective_projector_P}
\end{equation*}
where
\(\mathbf{x}=(x_i)_{i\in\Lambda }\)
and
\(\mathbf{y}=(y_i)_{i\in \Lambda}.\)
Since the corrected state of the selected subset is
\(\ket{\psi^-_n}^{\otimes k},\)
the joint probability of obtaining the outcome strings
\(\mathbf a=(a_i)_{i\in \Lambda}\)
and
\(\mathbf b=(b_i)_{i\in \Lambda}\)
is given by
\begin{align*}
&p(\mathbf a,\mathbf b
\mid \Lambda,\mathbf x,\mathbf y,\mathcal S)
\nonumber =
\operatorname{Tr}
\left[
\left(
M_{\mathbf a}^{\Lambda,\mathbf x}
\otimes
N_{\mathbf b}^{\Lambda,\mathbf y}
\right)
\ket{\psi^-}\bra{\psi^-}^{\otimes k}
\right]
\nonumber\\
&\quad =
\prod_{i\in \Lambda}
\operatorname{Tr}
\left[
\left(
\Pi_{a_i}^{P_{x_i}}
\otimes
\Pi_{b_i}^{V_{y_i}}
\right)
\ket{\psi^-}\bra{\psi^-}
\right]\\&\quad =4^{-k}
\prod_{i\in \Lambda}
\left[
1+(-1)^{a_i+b_i}E_{x_i y_i}
\right].
\label{eq:joint_subset_probability}
\end{align*}

Where for each selected pair \(i\in \Lambda\),
\(
E_{x_iy_i}=\left\langle P_{x_i}\otimes V_{y_i}\right\rangle_{\ket{\psi^-}}.
\)\\
Particularly for any two strings $\mathcal S,\mathcal S'$,
\begin{equation}
\label{eq:probability_secret_independence}
p(\mathbf a,\mathbf b
\mid \Lambda,\mathbf x,\mathbf y,\mathcal S)
=
p(\mathbf a,\mathbf b
\mid \Lambda,\mathbf x,\mathbf y,\mathcal S').
\end{equation}

Let $V^*$ be an arbitrary quantum-capable verifier with an auxiliary
register $E$, initially
uncorrelated with the shared Bell states,i.e.,
\begin{equation*}
\label{eq:initial_auxiliary}
\rho_{PVE}^{\mathcal S}
=
\ket{\Phi_{\mathcal S}}\bra{\Phi_{\mathcal S}}
\otimes\sigma_E,
\end{equation*} 
For fixed \((\Lambda,\mathbf{x},\mathbf{y})\), let \(\mathcal W_{\Lambda,\mathbf{x},\mathbf{y}}\) denote an arbitrary CPTP map applied by \(V^*\) prior to receiving the prover's response, and \(\mathcal R_{\Lambda,\mathbf{x},\mathbf{y},\mathbf{a}}\) an arbitrary CPTP map applied thereafter. For fixed \((\Lambda,\mathbf{x},\mathbf{y},\mathbf{a},\mathbf{b})\), let \(\rho_{VE}^{\Lambda,\mathbf{x},\mathbf{y},\mathbf{a},\mathbf{b},\mathcal S}\) be the corresponding normalized conditional state of the verifier's quantum registers immediately before \(\mathcal W\). Since the conditional state and probability are independent of the secret \(\mathcal S\), and CPTP maps preserve equality, the resulting verifier state remains independent of \(\mathcal S\).

\begin{align}
&
\mathcal R_{\Lambda,\mathbf x,\mathbf y,\mathbf a}
\circ
\mathcal W_{\Lambda,\mathbf x,\mathbf y}
\left[
\rho_{VE}^{\Lambda,\mathbf x,\mathbf y,\mathbf a,\mathbf b,\mathcal S}
\right]
\nonumber\\
&\qquad =
\mathcal R_{\Lambda,\mathbf x,\mathbf y,\mathbf a}
\circ
\mathcal W_{\Lambda,\mathbf x,\mathbf y}
\left[
\rho_{VE}^{\Lambda,\mathbf x,\mathbf y,\mathbf a,\mathbf b,\mathcal S'}
\right].
\label{eq:CPTP_preserves_independence}
\end{align}

To obtain the classical transcript, we dephase the transcript registers \(T=(\Lambda,\mathbf{x},\mathbf{y},\mathbf{a},\mathbf{b})\), corresponding to measurement in their computational bases while retaining the outcomes. Crucially, the dephasing acts as \(\mathcal D_T\otimes\mathcal I_{VE}\), leaving the verifier's quantum registers \(VE\) fully coherent. The resulting complete verifier view is therefore

\begin{align*}
\rho_{\mathrm{real}}^{V^*}(\mathcal S)
&=
\frac{1}{\binom Nk\,4^k}
\sum_{\substack{\Lambda\in\tilde{\Lambda}_k\\
\mathbf x,\mathbf y}}
\sum_{\mathbf a,\mathbf b}
p(\mathbf a,\mathbf b
\mid \Lambda,\mathbf x,\mathbf y,\mathcal S)
\nonumber\\
&\quad \otimes
\ket{\Lambda,\mathbf x,\mathbf y,\mathbf a,\mathbf b}
\bra{\Lambda,\mathbf x,\mathbf y,\mathbf a,\mathbf b}
\nonumber\\
&\quad\otimes
\mathcal R_{\Lambda,\mathbf x,\mathbf y,\mathbf a}
\circ
\mathcal W_{\Lambda,\mathbf x,\mathbf y}
\left[\rho_{VE}^{\Lambda,\mathbf x,\mathbf y,\mathbf a,\mathbf b,\mathcal S}
\right].
\label{eq:real_view_subset}
\end{align*}

Using Eqs.~\eqref{eq:probability_secret_independence} and
\eqref{eq:CPTP_preserves_independence}, we obtain
\begin{equation}
\rho_{\mathrm{real}}^{V^*}(\mathcal S)
=
\rho_{\mathrm{real}}^{V^*}(\mathcal S')
\label{eq:perfect_ZK_subset}
\end{equation}
for every pair of secret strings $\mathcal S$ and $\mathcal S'$.\\

\noindent\textbf{Simulator Super-operator:} The simulator does not need to know the actual secret string \(\mathcal S\), instead it proceeds as follows:
\begin{enumerate}
 \item Choose a fixed reference string
 \(
 \mathcal S_0=(0,\ldots,0)
 \)
 and prepare
 \[
 \rho_{PVE}^{\mathcal S_0}
 =
 |\beta_{\mathcal S_0}\rangle\langle\beta_{\mathcal S_0}|
 \otimes\sigma_E,
 \quad
 |\beta_{\mathcal S_0}\rangle
 =
 |\psi^-\rangle^{\otimes N}.
 \]

 \item Interacting with the verifier it obtains the verifier's challenge
 \(\Lambda\in\tilde{\Lambda}_k,~~ |\Lambda|=k,
 \)
 through the prescribed classical subset-selection interface. For every selected \(i\in \Lambda\), it obtains the CHSH measurement setting pair
 \((x_i,y_i).
 \)
 \item The simulator then internally executes the verifier $V^*$ on this reference
 state. For every selected subset $\Lambda$, it performs the reference
 corrections and CHSH measurements and generates the corresponding prover
 responses.
\end{enumerate}
The simulator's output state is $\rho_{\mathrm{sim}}^{V^*}$ which is equal to $\rho_{\mathrm{real}}^{V^*}(\mathcal S_0)$ and by Eq.~\eqref{eq:perfect_ZK_subset}, 
\begin{equation*}
\label{eq:simulator_output}
\rho_{\mathrm{sim}}^{V^*}
=
\rho_{\mathrm{real}}^{V^*}(\mathcal S_0)=\rho_{\mathrm{real}}^{V^*}(\mathcal S)
\end{equation*}
for every secret string $\mathcal S$. Thus Protocol \ref{pro_3} is perfect ZK against arbitrary malicious quantum-capable verifiers with a classical communication interface.

\subsection{Simulator for Protocol \ref{pro_m2}}
\noindent\textbf{Transcript Super-operator:} Consider an ensemble of $N$ shared $n$-qubit states, each independently chosen from the generalized GHZ set $\{\lvert\psi_n^\pm\rangle, \lvert\phi_n^\pm\rangle\}$the corresponding joint state shared among the prover and the verifiers is
\[
\lvert\Psi_\gamma\rangle = \bigotimes_{r=1}^N \lvert\gamma_r\rangle, \quad \lvert\gamma_r\rangle \in {\lvert\psi_n^\pm\rangle, \lvert\phi_n^\pm\rangle}.
\]

Upon receiving the challenge index $r$, the prover applies a local unitary correction $U_{\gamma_r}$ conditioned on the hidden identity of the $r$-th system. By construction, each selected state maps to the canonical GHZ target $\lvert\phi_n^-\rangle = (\lvert0\rangle^{\otimes n} - \lvert1\rangle^{\otimes n})/\sqrt{2}$ up to an irrelevant global phase. Consequently, the verifier's subsequent measurement statistics are strictly invariant under the initial state label. Encoding the challenge index $r \in [N]$ into a register $C_R$, its coherent purification is,
\(
\lvert\text{chal}\rangle = \frac{1}{\sqrt{N}} \sum_{r=1}^N \lvert r\rangle_{C_R}.
\)
For the generalized $(n+1)$-qubit generalized GHZ operators $\{\hat{\mathcal{O}}_i\}_{i=0}^n$ defined in Eq. \eqref{operator}, the measurement settings state, stored in the register $c_I$ is, \(|\mathrm{set}\rangle = \frac{1}{\sqrt{n+1}}\sum_{i=0}^n |i\rangle_{c_I}.\) Conditioned on setting $i$, party $j$ measures the local observable\begin{equation*}A_j^{(i)} =\begin{cases}Y_j, & i \neq 0 \text{ and } j \in \{i, i+1\},\\ X_j, &\text{otherwise},\end{cases}\end{equation*} with $n+1 \equiv 1 \pmod n$. The corresponding projection operator for party $j$ yielding outcome $a_j \in \{0,1\}$ is given by $\Pi_{a_j}^{(i)} = \frac{1}{2}[I + (-1)^{a_j} A_j^{(i)}]$. For the $r$-th shared system, the corresponding measurement operator is\begin{equation*}M_{a_j}^{(r,i)} = I^{\otimes (r-1)} \otimes \Pi_{a_j}^{(i)} \otimes I^{\otimes (N-r)}.\end{equation*}
The complete measurement outcome is denoted by
 $\mathbf{a} = (a_1, \dots, a_n) \in \{0,1\}^n$ . The corresponding $n$-party measurement operator can be written as
\[
M_{\mathbf a}^{(r,i)}
=
\bigotimes_{j=1}^{n}
M_{a_j}^{(r,i)}.
\]
To define the classical transcript, let
\(
A=(A_1,\ldots,A_n)
\)
denote the classical outcome register. The classicalization of the challenge and outcome registers is represented by the dephasing map
\[
\mathcal D_{C_RC_IA}(\rho)
=
\sum_{r,i,\mathbf a}
\left\lvert ri\mathbf a\right\rangle
\left\langle ri\mathbf a\right\rvert
\rho
\left\lvert ri\mathbf a\right\rangle
\left\langle ri \mathbf a\right\rvert.
\]
The classical transcript state of the real protocol is therefore
\begin{equation*}
\rho_{\mathrm{real}} = \frac{1}{N(n+1)} \sum_{r,i,\mathbf{a}} p(\mathbf{a} \mid i, r) |r\, i\, \mathbf{a}\rangle \langle r\, i\, \mathbf{a}|.
\end{equation*}
In an honest execution, each challenged state is transformed by the prover into the canonical state $\ket{\phi_n^-}$. The conditional outcome distribution is , $p(\mathbf{a} \mid i, r) = p_{\ket{\phi_n^-}}(\mathbf{a} \mid i)$. In particular, the distribution is independent of both the challenged index $r$ and the hidden state identity $\gamma_r$.
 and satisfies the GHZ parity constraints \eqref{eigen_value} i.e.,\begin{equation*}p_{\ket{\phi_n^-}}(\mathbf{a} \mid i) =\begin{cases}2^{-(n-1)}, & \bigoplus_{j=1}^n a_j = \delta_{i,0} \\ 0,& \text{otherwise}.\end{cases}\end{equation*} for $i \in \{0, \dots, n\}$, where $\delta_{i,0}$ is the Kronecker delta.

 If the verifier possesses an auxiliary private quantum register $E$, initially in state $\sigma_E$ and independent of the shared GHZ states,then the real transcript state has the form 
\begin{equation*}
\begin{aligned}
\rho_{\text{real}} &= \frac{1}{N(n+1)2^{n-1}} \sum_{r=1}^{N} \Big[ \sum_{\oplus_j a_j = 1} |r\,0\,\mathbf{a}\rangle\langle r\,0\,\mathbf{a}| \\&\quad + \sum_{i=1}^{n} \sum_{\oplus_j a_j = 0} |r\,i\,\mathbf{a}\rangle\langle r\,i\,\mathbf{a}| \Big] \otimes \sigma_E. 
\end{aligned}
\end{equation*}

\noindent\textbf{Simulator Super-operator:} The simulator prepares the corresponding reference state
\[
\lvert\Psi_{\boldsymbol{\gamma}_0}\rangle
=
\lvert\phi_n^-\rangle^{\otimes N}.
\]
It samples the challenged index $r$ and the measurement setting $i$ uniformly, and coherently from the verifier. Since the reference state is already the canonical state, the simulator does not need the hidden identity of the challenged state. Interacting with the verifiers, it performs the prescribed prover-side measurement on the challenged copy and obtains the corresponding outcome. Together with the outcomes generated by the verifiers, this reproduces the ideal distribution,
\begin{equation*}p_{\mathrm{sim}}(\mathbf{a}\mid r, i) = p_{|\phi_n^-\rangle}(\mathbf{a} \mid i) = p_{\mathrm{real}}(\mathbf{a}\mid r, i).\end{equation*} Since the verifiers’ arbitrary private auxiliary register \(E\), initially described by an arbitrary density operator \(\sigma_E\), is independent of the hidden-state identity, the simulator initializes its auxiliary register in the same state \(\sigma_E\). Finally, the simulator applies the same classical dephasing operation
\(
\mathcal D_{C_RC_SA}
\)
to the simulated transcript. The resulting simulator state is therefore, 
\begin{equation*}\rho_{\mathrm{sim}} = \frac{1}{N(n+1)} \sum_{r,i,\mathbf{a}} p_{\mathrm{sim}}(\mathbf{a}\mid r, i) |r\, i\, \mathbf{a}\rangle \langle r\, i\, \mathbf{a}| \otimes \sigma_E,\end{equation*} which satisfies $\rho_{\mathrm{real}} = \rho_{\mathrm{sim}}$. Because $\rho_{\mathrm{sim}}$ is generated without reference to the true state identity while matching $\rho_{\mathrm{real}}$ exactly, the classical transcript achieves perfect zero-knowledge.

\section{Self-testing of GHZ correlations}
Recently, Das et al. \cite{Das2025} introduced a self-testing scheme for $n$-qubit GHZ correlations based on the Bell operator
\be\label{GHZ_Bell}\displaystyle \mathcal{B}_n=-\hat{\mathcal{O}}_0+\sum_{i=1}^{n}\hat{\mathcal{O}}_i\ee
constructed from the eigenvalue relations in Eq. \eqref{eigen_value}. For odd $n$, local-realistic (LR) constraint provides the bound
\[
|\langle\mathcal{B}_n\rangle_{\rm LR}|\leq n-1,
\]
whereas quantum correlations allow to attain the algebraic maximum,
\begin{equation}
\label{max_bell}
\beta_Q=|\langle\mathcal{B}_n\rangle_{\rm Q}|=n+1.
\end{equation}
Moreover, saturation at $\beta_Q=\pm(n+1)$ self-tests the GHZ states $\ket{\phi_n^\mp}$, respectively, up to local unitaries.

Here, we extend this approach to self-test not only $\ket{\phi_n^\pm}$ but the complete GHZ basis ${\ket{\phi_n(r)^\pm}}_{r=1}^{n}$, where
\ben
\ket{\phi_n(r)^\pm}
=\frac{1}{\sqrt{2}}\left(
\ket{0}^{r}\ket{1}^{n-r}
\pm
\ket{1}^{r}\ket{0}^{n-r}
\right).\een
Our approach differs fundamentally from conventional inequality-based self-testing \cite{Das2025}: we require neither a Bell inequality nor a maximal-violation condition. Instead, we establish a multipartite logical no-go argument based on the set of eigenvalue relations in Eq. \eqref{eigen_value}. In contrast to self-testing based on a single maximal Bell-violation constraint, our construction employs multiple simultaneous constraints to uniquely characterize the target GHZ correlations. This structure also facilitates a direct security analysis of our zero-knowledge proof (ZKP) schemes against most general coherent attacks. We state our main results about self-testing in the following two theorems.
\begin{thm}[GHZ correlation self-testing] \label{th}
 For odd $n$, the eigenvalue relations in Eq. \eqref{eigen_value} i.e., \be\label{GHZ_P}
 \begin{aligned}
 \langle \hat{\mathcal{O}}_0\rangle_\rho&=-1\\
 \langle \hat{\mathcal{O}}_i\rangle_\rho&=+1, \forall i\in[n]\\
 \end{aligned}
 \ee self-test the corresponding correlations \(\rho=\proj{\phi_n^-}\), where \(\ket{\phi_n^-}=\frac{1}{\sqrt{2}}\left(\ket{0}^{\otimes n}-\ket{1}^{\otimes n}\right)\).\end{thm}
\begin{proof}
 The proof of the theorem follows immediately from Lemmas \ref{lem1}-\ref{lem2}.
\end{proof}
\begin{lem}\label{lem1}
 For odd $n$, the eigenvalue relations in Eq. \eqref{eigen_value} are incompatible with any local-realistic (LR) correlations and uniquely identify $\ket{\phi_n^-}$, up to an irrelevant global phase.

\end{lem}

\begin{proof}
Consider the relations in Eq. \eqref{GHZ_P} collectively, with the observables defined in Eq. \eqref{operator}. For each site $i\in[n]$, $Y_i$ occurs twice, whereas $X_i$ occurs $n-1$ times, which is even for odd $n$. Assigning deterministic LR values $\pm1$ to all local observables and taking the product of all relations in Eq. \eqref{GHZ_P}, each assigned value on the left-hand side occurs an even number of times, yielding $+1$. By contrast, the product of the corresponding eigenvalues on the right-hand side is $-1$, leading to a contradiction. Since any probabilistic LR strategy is a convex combination of deterministic strategies, it cannot evade this contradiction. Hence, no LR correlations can satisfy all the relations in Eq. \eqref{GHZ_P} simultaneously.

To determine whether the relations in Eq. \eqref{GHZ_P} can be satisfied by quantum correlations, and in particular by a pure $n$-qubit state, consider the most general $n$-qubit pure state,
\begin{equation*}\label{pure_state}
|\psi\rangle=\sum_{b_1,\dots,b_n\in\{0,1\}}\alpha_{b_1\dots b_n}|b_1\dots b_n\rangle,
\sum|\alpha_{b_1\dots b_n}|^2=1.
\end{equation*}\label{gen_ghz}
Using \begin{align*}
 \hat{\mathcal{O}}_0|b_1\dots b_n\rangle&=|\bar b_1\dots \bar b_n\rangle,\\
 \hat{\mathcal{O}}_i|b_1\dots b_n\rangle&=(-1)^{1\oplus b_i\oplus b_{i+1}}|\bar b_1\dots \bar b_n\rangle,
\end{align*} the first eigenvalue relation of Eq. \eqref{eigen_value} implies $\alpha_{b_1,\dots,b_n}=-\alpha_{\bar{b}_1,\dots,\bar{b}_n}$. Substituting this relation into the remaining eigenvalue equations yields $\alpha_{b_1,\dots,b_n}=(-1)^{b_i\oplus b_{i+1}}\alpha_{b_1,\dots,b_n}$ for all $i\in[n]$. Hence, $\alpha_{b_1\ldots b_n}$ can be nonzero only when $b_i=b_{i+1}$ for every $i$, restricting the support of $\ket{\psi}$ to $\ket{0}^{\otimes n}$ and $\ket{1}^{\otimes n}$. Together with the first eigenvalue relation and normalization, this uniquely gives \(\ket{\psi}=e^{-i\gamma }\ket{\phi_n^-}\) and therefore $\rho=\proj{\phi_n^-}$. 

\end{proof}

The device-independent analysis relies on the following variant of Jordan's lemma \cite{masaness}.

\begin{lem}\label{lem2}
Consider two dichotomic observables \(A\) and \(B\) acting on a Hilbert space \(\mathscr{H}\).
Then \(\mathscr{H}\) admits an orthogonal decomposition
\[
\mathscr{H}=\bigoplus_{\lambda\in\Lambda}\mathscr{H}^{\lambda},
\qquad
\dim(\mathscr{H}_{\lambda})\leq 2,
\]
such that each sector \(\mathscr{H}^{\lambda}\) is simultaneously invariant under \(A\) and \(B\). Accordingly, the two observables can be expressed as
\[
A=\bigoplus_{\lambda\in\Lambda}A^{\lambda},
\qquad
B=\bigoplus_{\lambda\in\Lambda}B^{\lambda},
\]
where \(A^{\lambda}:=A_{|{\mathscr{H}^{\lambda}}}\) and \(B^{\lambda}:=B_{|{\mathscr{H}^{\lambda}}}\) denote their respective restrictions to the invariant sector \(\mathscr{H}^{\lambda}\).
\end{lem}

For each invariant sector \(\mathcal{H}^{\lambda}\), let \(P^{\lambda}\) denote the associated projector. The physical observables \(X\) and \(Y\) then induce the sector-restricted operators

$$
X^{\lambda}:=P^{\lambda}XP^{\lambda},
\qquad
Y^{\lambda}:=P^{\lambda}YP^{\lambda}.
$$

These operators fully characterize the action of the corresponding untrusted measurements within the sector \(\mathcal{H}^{\lambda}\).

Under this decomposition, the state \(\rho\) and the observables \(\{\hat{\mathcal O}_i\}_{i=0}^{n}\) in Eq.~\eqref{operator} acquire the block-diagonal forms
\begin{equation*}
\rho=\bigoplus_{\boldsymbol{\nu}}
p_{\boldsymbol{\nu}}\rho^{\boldsymbol{\nu}}\quad\text{and}\quad
\hat{\mathcal O}_i=
\bigoplus_{\boldsymbol{\nu}}
\hat{\mathcal O}^{\boldsymbol{\nu}}_i,
\end{equation*} where \(\boldsymbol{\nu}=(\nu_1,\ldots,\nu_n)\). The corresponding block operators are
\begin{equation*}
\hat{\mathcal O}^{\boldsymbol{\nu}}_0
=X^{\nu_1}\cdots X^{\nu_n},
\end{equation*}
and, for \(i\in[n]\),
\begin{equation*}
\hat{\mathcal O}^{\boldsymbol{\nu}}_i
=X^{\nu_1}\cdots X^{\nu_{i-1}}
Y^{\nu_i}Y^{\nu_{i+1}}
X^{\nu_{i+2}}\cdots X^{\nu_n}.
\end{equation*}
It follows that
\begin{equation*}
\langle\hat{\mathcal O}_i\rangle_{\rho}=
\sum_{\boldsymbol{\nu}}p_{\boldsymbol{\nu}}
\langle\hat{\mathcal O}^{\boldsymbol{\nu}}_i\rangle_{\rho^{\boldsymbol{\nu}}},
\quad
\langle\hat{\mathcal O}^{\boldsymbol{\nu}}_i\rangle_{\rho^{\boldsymbol{\nu}}}
\equiv
\operatorname{Tr}\left(
\rho^{\boldsymbol{\mu}}\hat{\mathcal O}^{\boldsymbol{\nu}}_i
\right),
\end{equation*}
for \(i=0,\ldots,n\), where \(\boldsymbol{\nu}=(\nu_1,\ldots,\nu_n)\). Because each \(\hat{\mathcal O}^{\boldsymbol{\nu}}_i\) has eigenvalues in \(\{\pm1\}\), the extremal conditions
\begin{equation*}
\langle\hat{\mathcal O}_0\rangle_{\rho}=-1,
\qquad
\langle\hat{\mathcal O}_i\rangle_{\rho}=+1,
\quad i\in[n],
\end{equation*}
can be attained only if every occupied block saturates the same bounds, namely
\begin{equation*}
\langle\hat{\mathcal O}^{\boldsymbol{\nu}}_0\rangle_{\rho^{\boldsymbol{\nu}}}=-1,
\qquad
\langle\hat{\mathcal O}^{\boldsymbol{\nu}}_i\rangle_{\rho^{\boldsymbol{\nu}}}=+1,
\quad i\in[n].
\end{equation*}
Lemma \ref{lem1} then fixes each nonzero block to the unique common eigenstate, \(\rho^{\boldsymbol{\nu}}=\ket{\phi_n^-}
\bra{\phi_n^-}^{\boldsymbol{\nu}}\) and hence
\begin{equation*}
\rho=\ket{\Psi}\bra{\Psi},
\quad
\ket{\Psi}=
\bigoplus_{\boldsymbol{\nu}}
\sqrt{p_{\boldsymbol{\nu}}},
\ket{\phi_n^-}^{\boldsymbol{\nu}}.
\end{equation*}

To extract the ideal state, let each party append an ancillary qubit in \(\ket{0}'_j\) and implement the local isometry
\begin{equation*}
V_j:
\ket{\tau}^{\nu_j}\ket{0}'_j
\longmapsto
\ket{0}^{\nu_j}\ket{\tau}'_j,
\quad
\tau\in\{0,1\}.
\end{equation*}
The global isometry \(V=\bigotimes_{j=1}^{n}V_j\) consequently maps
\begin{equation*}
\ket{\Psi}\otimes\ket{0}'^{\otimes n}
\longmapsto
\ket{\xi}\otimes\ket{\phi_n^-}',
\end{equation*}
where \(\ket{\phi_n^-}'\) is the ideal \(n\)-qubit GHZ state on the ancillary systems and \(\ket{\xi}\) contains the residual sector degrees of freedom. Thus, the target GHZ state is locally extractable, establishing the desired self-testing statement.

The preceding analysis readily extends Theorem~\ref{th} to the complete set of $n$-qubit GHZ-basis states.

\begin{thm}\label{th2}
For an odd integer $n$, the Greenberger-Horne-Zeilinger (GHZ)-type paradox defined by the eigenvalue relations
\begin{equation*}
\label{eq:ghz_paradox}\begin{aligned}
\hat{\mathcal{O}}_0 |\psi_n(\Lambda)^\pm\rangle &= \pm |\psi_n(\Lambda)^\pm\rangle,\\
\hat{\mathcal{O}}_i |\psi_n(\Lambda)^\pm\rangle &= 
\begin{cases}
\mp |\psi_n(\Lambda)^\pm\rangle, & \text{if } \{i, i+1\} \subseteq \Lambda \\
&\text{ or } \{i, i+1\} \subseteq \bar{\Lambda}, \\
\pm |\psi_n(\Lambda)^\pm\rangle, & \text{otherwise},
\end{cases} 
\end{aligned}
\end{equation*}
self-tests the generalized $n$-qubit GHZ state
\begin{equation*}
\label{eq:target_state}
|\psi_n(\Lambda)^\pm\rangle = \frac{1}{\sqrt{2}} \left( |0\rangle^{\Lambda} |1\rangle^{\bar{\Lambda}} \pm |1\rangle^{\Lambda} |0\rangle^{\bar{\Lambda}} \right),
\end{equation*}
where $\Lambda \subseteq [n]$ and $\bar{\Lambda} = [n] \setminus \Lambda$ denotes its complement.
\end{thm}

\section{Robustness under noise}
\label{sec:robust}
The objective of the zero-knowledge proof (ZKP) is to enable a prover to convince the verifiers that they possess the identity of a shared quantum state without revealing the state itself. We assume that the verifiers know only the candidate ensemble and have no prior knowledge of the actual state; verification against a known state is beyond the scope of this work. The shared states are assumed to be prepared either by the prover or by an independent dealer. Soundness therefore reduces to bounding the probability that a dishonest prover, possessing incomplete or no information about the state, can falsely certify complete knowledge.

The key observation is that limited knowledge of the state constrains the prover's ability to reproduce the correlations required to win the multipartite GHZ game. Equivalently, the prover's guessing probability for the verifiers' joint outcomes is constrained by the observed GHZ-winning deficiency, or, equivalently, by the corresponding Bell-inequality violation. Thus, a bound on the maximal GHZ-game winning probability directly yields a bound on the verifiers' joint-outcome guessing probability, and vice versa. We first establish the soundness bound for a single round under an i.i.d.\ assumption and then extend it to arbitrary coherent attacks using the entropy accumulation theorem (EAT) \cite{Arnon-Friedman2018,Tomamichel2009AEP}. Since the adversarial prover constitutes the primary security threat, we conservatively allow any external eavesdropper to possess arbitrary side information and to collaborate fully with the prover.

\subsection{GHZ paradox under noise}

Any $n$-qubit density operator $\rho$ ($\rho \ge 0$, $\mathrm{Tr}[\rho] = 1$) can be expanded in the Pauli basis as\begin{equation}\label{eq:rho_pauli}\rho = \frac{1}{2^n} \sum_{i_1, \dots, i_n = 0}^{3} T_{i_1 \dots i_n} \big( \sigma_{i_1} \otimes \dots \otimes \sigma_{i_n} \big),\end{equation}where $\sigma_0 = \mathbb{I}_2$, $\{\sigma_1, \sigma_2, \sigma_3\} \equiv \{\sigma_x, \sigma_y, \sigma_z\}$, and the real correlation tensor elements are $T_{i_1 \dots i_n} = \mathrm{Tr}[\rho \, (\sigma_{i_1} \otimes \dots \otimes \sigma_{i_n})]$, with normalization $T_{0\dots 0} = 1$. Under a local Pauli conjugation on the first qubit, $\rho_m = (\sigma_m \otimes \mathbb{I}^{\otimes (n-1)}) \rho (\sigma_m \otimes \mathbb{I}^{\otimes (n-1)})$ with $m \in \{1,2,3\}$, the state retains the form of Eq.~\eqref{eq:rho_pauli} with transformed coefficients:\begin{equation*}T^\prime_{i_1 \dots i_n} = (-1)^{1 - \delta_{0 i_1} - \delta_{m i_1}+\delta_{m0}} T_{i_1 \dots i_n}.\end{equation*}This sign rule directly reflects the algebraic relation $\sigma_m \sigma_{i_1} \sigma_m = (-1)^{1 - \delta_{0 i_1} - \delta_{m i_1}+\delta_{m0}} \sigma_{i_1}$: components with $i_1 = 0$ (identity) or $i_1 = m$ commute and remain invariant, while the two orthogonal non-identity Pauli components anticommute and pick up a factor of $-1$. After receiving the challenge index $m$ from the verifiers, the prover applies $\sigma_m$ before revealing any measurement outcomes. At this stage, the verifiers have no information about $m$. On the other hand, a prover with negligible knowledge of the identity of the shared state can satisfy the GHZ paradox with only negligible probability. We therefore restrict our analysis to the state $\rho$. To characterize the state $\rho$ satisfying simultaneously the GHZ eigenvalue relations \eqref{GHZ_P} within an $\epsilon\geq 0$-tolerance, we consider the correlation thresholds:
\be \label{noise}\begin{aligned} 
\langle \hat{\mathcal{O}}_0 \rangle_{\rho} &\leq -1 + \epsilon, \\
\langle \hat{\mathcal{O}}_s \rangle_{\rho} &\geq 1 - \epsilon \quad (\forall s \in [n]),
\end{aligned}\ee
which in terms of the correlation tensor elements translates to
\ben\label{noise_1}
\begin{aligned}
-T_{11\dots 1} &\geq 1 - \epsilon, \\
T_{1\dots 1221\dots 1} &\geq 1 - \epsilon,\\T_{211\dots 12} &\geq 1 - \epsilon.
\end{aligned}\een
Defining the stabilizer generators $S_0 = -\hat{\mathcal{O}}_0$ and $S_s = \hat{\mathcal{O}}_s$ for $s \in [n]$, we observe that $S_j^2 = \mathbb{I}^{\otimes n}$ and $[S_j, S_k] = 0$. The set of $n$ independent generators $\{S_j\}_{j=0}^{n-1}$ forms an Abelian stabilizer group $\mathcal{S}$ of order $\vert{}\mathcal{S}\vert{} = 2^n$, whose $+1$ common eigenspace is one-dimensional and uniquely specifies the pure state $\vert{}\phi_n^-\rangle$. The corresponding rank-one projector is obtained via the group average:
\begin{equation*}\begin{aligned}
|\phi_n^-\rangle \langle \phi_n^-| &= \frac{1}{2^n} \sum_{S \in \mathcal{S}} S \\&
= \frac{1}{2^n} \left[ \mathbb{I}^{\otimes n} + \sum_{j=0}^{n-1} S_j + \sum_{|\Lambda| \ge 2} \prod_{k \in \Lambda} S_k \right], 
\end{aligned}
\end{equation*}
where the third term sums all higher-order products over subsets $\Lambda \subseteq \{0, \dots, n-1\}$, completing the summation over all group elements.

By robust self-testing, any valid density matrix meeting the threshold conditions in Eq.~\eqref{noise} can be written in generic perturbed form as\begin{equation}\label{rho_decomp}\rho = (1 - \eta) |\phi_n^-\rangle \langle \phi_n^-| + \eta \rho_\perp,\end{equation}where $\rho_\perp$ is a valid density matrix supported on the orthogonal complement $\mathrm{supp}(\rho_\perp) \subseteq \mathrm{span}\{\vert{}\Phi_n^-\rangle\}^\perp$, and the noise parameter satisfies $\eta \le \epsilon/2$. Consequently, $\rho$ obeys the corresponding Bell inequality $\mathrm{Tr}[\rho \mathcal{B}_n] \geq (n+1)(1 - \epsilon)$, or equivalently,\ben\begin{aligned}
-T_{1\dots 1}+T_{221\dots 1}+T_{1221\dots 1}&+T_{1\dots 122}+T_{211\dots 12}\\&\geq (n+1)\epsilon,
 \end{aligned}
\een subject to trace normalization $T_{0\dots 0} = 1$ and global positivity $\rho \ge 0$. Thus we have the following lemma. 
\begin{lem}\label{lemma1}
 If a state $\rho$ satisfies the GHZ paradox with a noise tolerance $\epsilon \ge 0$ i.e., simultaneously satisfies Eq.~\eqref{noise}, then it necessarily admits the decomposition given in Eq.~\eqref{rho_decomp}.
\end{lem}

In particular, the isotropic Werner-like state satisfying these bounds takes the form\begin{equation*}\label{werner}\rho_{\mathrm{W}} = (1 - \eta) |\phi_n^-\rangle \langle \phi_n^-| + \frac{\eta}{2^n} \mathbb{I}^{\otimes n},\end{equation*}where the non-identity correlation tensor elements associated with the stabilizer group evaluate to $T_{i_1 \dots i_n} = (1 - \eta) \, \langle \phi_n^-\vert{} \sigma_{i_1} \dots \sigma_{i_n} \vert{}\phi_n^-\rangle$, while all remaining coefficients vanish identically.

\subsection{ZKP security against coherent attack}

For the GHZ-paradox-based ZKP Protocol \ref{pro_m2}, the entropy accumulation theorem (EAT) \cite{Arnon-Friedman2018} provides a natural framework for lifting single-round entropy bounds to sequential, multi-round executions against coherent quantum adversaries. Because the EAT itself does not prescribe the single-round entropy rate as a function of Bell violation, this relation must be certified independently via device-independent semi-definite programming (SDP) hierarchies. To establish composable zero-knowledge security against arbitrary, coherent quantum attacks, we formulate the $N$-round execution of Protocol~\ref{pro_m2} using the Entropy Accumulation Theorem \ref{EAT_th}. A dishonest prover may correlate all $N$ rounds across arbitrary quantum side information $E$. In round $i \in [N]$, let $X_i \in \{0, \dots, n\}$ be the verifier's uniform GHZ challenge and $C_i \in \{0, 1\}$ denote the binary pass/fail indicator for the parity constraints $\langle \hat{\mathcal{O}}_0 \rangle = -1$ and $\langle \hat{\mathcal{O}}_j \rangle = 1$ ($\forall j \in [n]$). The protocol accepts provided the total accepted rounds $F = \sum_{i=1}^N C_i$ satisfy $F/N \ge 1 - \gamma$, bounding the observed winning frequency by $q_{\mathrm{obs}} \ge 1 - \gamma$.

Following the entropy accumulation framework of \cite{Arnon-Friedman2018}, we first define the EAT channel associated with the sequential implementation of our protocol.\\

\noindent\textbf{EAT channels:} The sequential verification is modeled by completely positive trace-preserving (CPTP) maps $$\mathcal{M}_i: R_{i-1} \to R_i O_i S_i C_i,$$ for $i \in [N]$, where $R_i$ is the unmeasured prover state forwarded to round $i+1$,
$O_i$ denotes the round-$i$ output relevant to the verification test,
$S_i$ is the corresponding quantum side-information register, and $C_i$
is a finite-dimensional classical register containing the outcome of the
GHZ verification test. In the present protocol $ C_i\in\{0,1\}$,
where $C_i=1$ denotes acceptance and $C_i=0$ denotes rejection of the
GHZ constraint.

For every $i\in[N]$, the maps $\mathcal{M}_i$ satisfy the following
properties.

\begin{enumerate}
 \item The registers $C_i$ are classical, while $R_i$, $O_i$, and
 $S_i$ are quantum registers. We denote the dimension of $O_i$ by
 $d_{O_i}$.

 \item Let $R'\simeq R_{i-1}$ be an auxiliary register isomorphic to
 the incoming memory. For every state
 $\tau_{R_{i-1}R'}$, define
 \begin{equation*}
 \sigma_{R_iO_iS_iC_iR'}
 =
 (\mathcal{M}_i\otimes\mathbb{I}_{R'})
 \left(\tau_{R_{i-1}R'}\right).
 \end{equation*}
 The register $C_i$ is classical and its value can be obtained by a
 measurement on $O_iS_i$ without disturbing the corresponding
 post-measurement state.

 \item Let $E$ denote an arbitrary quantum environment held by the
 prover or adversary, including any purification of the initial
 shared state and any quantum side information retained throughout
 the protocol. For any initial joint adversary-prover state $\rho_{R_0 E}$, the global state after $N$ rounds,\bean &\rho_{O_{1\sim N} S_{1\sim N} C_{1\sim N} E}&\\&\qquad= \mathrm{Tr}_{R_N}\left[ \left( \mathcal{M}_N \circ \dots \circ \mathcal{M}_1\right) \otimes \mathcal{I}_E \right] \rho_{R_0 E}, & 
 \eean where notation $X_{1\sim N}$ denotes $X_{1\sim N}=X_1,\dots ,X_N$, satisfies the Markov chain condition $O_{1\sim (i-1)} \leftrightarrow S_{1\sim (i-1)}E \leftrightarrow S_i$ for each $i \in [N]$, equivalent to the vanishing conditional mutual information: \be I(O_{1\sim (i-1)} : S_i \mid S_{1\sim (i-1)}E)_\rho = 0. \label{eq:markov}\ee \end{enumerate}
 
 The last condition is the structural requirement that permits entropy
accumulation in the presence of arbitrary inter-round quantum
correlations. In particular, the registers $R_i$ may retain quantum
memory between successive rounds, so that the above formulation does
not impose an independent-and-identically-distributed assumption on the
GHZ tests.
 
 Condition (\ref{eq:markov}) holds naturally in the ZKP setting because round-by-round side information leakage $S_i$ depends only on the current query inputs and provers' internal state, shielded from historical challenge-answer transcripts given prior side information and initial entanglement $E$.

 To quantify single-round entropic rates, let $\mathbf{p} \in \mathcal{P}(\mathcal{C})$ be a distribution over the test alphabet $\mathcal{C}$. We define the restricted state space
 
 \begin{equation*}
\begin{aligned}
\Sigma_i(\mathbf{p})
=
\Big\{
&\sigma_{O_i S_i C_i R_i R'}
=
(\mathcal{M}_i\otimes\mathcal{I}_{R'})
\bigl(\tau_{R_{i-1}R'}\bigr)
\;\Big|\;\\
&\tau_{R_{i-1}R'}\in
\mathcal{D}(R_{i-1}\otimes R'),
\quad
\sigma_{C_i}=\mathbf{p}
\Big\},
\end{aligned}
\end{equation*}

 where $\sigma_{C_i} = \sum_{c \in \mathcal{C}} \langle c \vert{} \sigma_{C_i} \vert{} c \rangle \vert{}c\rangle\langle c\vert{}$.

 \begin{defn} {Min-tradeoff function -} A continuous function $f_{\min}: \mathcal{P}(\mathcal{C}) \to \mathbb{R}$ is a min-tradeoff function for $\mathcal{M}_i$ if$$f_{\min}(\mathbf{p}) \le \inf_{\sigma \in \Sigma_i(\mathbf{p})} H(O_i \mid S_i R')_\sigma,$$with $f_{\min}(\mathbf{p}) = +\infty$ when $\Sigma_i(\mathbf{p}) = \emptyset$, where $H(A \mid B)_\sigma$ is the conditional von Neumann entropy. 
\end{defn}
In the context of the $n$-qubit GHZ game, the winning probability directly constrains the distance of $\tau_{R_{i-1}R'}$ from ideal GHZ states via self-testing bounds. The convex function $f_{\min}(\mathbf{p})$ thereby establishes an operational lower bound on the generation rate of smooth min-entropy $H_{\min}^{\varepsilon}(O_{1\sim N} \mid S_{1\sim N} E)$ against malicious verifiers or provers, reducing the full security proof to characterizing single-round nonlocal violations.

\begin{thm}[EAT for $n$-Qubit GHZ ZKP] \label{EAT_th}
Let $\{\mathcal{M}_i\}_{i=1}^N$ be EAT channels with output dimension $\dim(O_i) = d_{O_i}$ for each round $i \in [N]$, and let $\varepsilon \in (0, 1)$ be the smoothing parameter. If $f_{\min}: \mathcal{P}(\mathcal{C}) \to \mathbb{R}$ is a convex min-tradeoff function for $\{\mathcal{M}_i\}$ satisfying $f_{\min}(\mathrm{freq}(c_{1\sim N})) \ge t$ for all $c_{1\sim N} \in \Omega$ with $\operatorname{Pr}[c_{1\sim N}]_{\rho_{\vert{}\Omega}} > 0$, then the smooth conditional min-entropy of the cumulative transcript $O_{1\sim N}$ is lower bounded by\be \label{eq:eat_min}H_{\min}^{\varepsilon}(O_{1\sim N} \mid S_{1\sim N} E)_{\rho_{|\Omega}} > N t - v \sqrt{N},\ee where the second-order finite-size penalty parameter $v$ is\ben \label{eq:v_param}v = 2\left(\log(1 + 2 d_{O_i}) + |\nabla f_{\min}|\infty \right) \sqrt{1 - 2\log(\varepsilon_s \cdot p_\Omega)}.\een
\end{thm}

Since, verification in Protocol~\ref{pro_m2} is governed by the eigenvalue constraints of Eq.~\eqref{GHZ_P} and the underlying Bell operator [Eq.~\eqref{GHZ_Bell}] is linear in the observed correlators, the proof follows directly from an extension of the entropy accumulation theorem of \cite{Arnon-Friedman2018}.

In the $n$-qubit GHZ-ZKP protocol, Eq.~(\ref{eq:eat_min}) provides the finite-round rate for extractor privacy and soundness: even if cheating provers share arbitrary entangled states across rounds, passing the GHZ nonlocal tests ($p_\Omega \approx 1$) certifies that the transcript $O_{1\sim N}$ contains at least $N t - \mathcal{O}(\sqrt{N})$ bits of smooth min-entropy independent of the adversary's quantum memory $S_{1\sim N} E$.

For the zero-knowledge protocol, soundness concerns the probability that a dishonest prover is accepted without possessing the claimed GHZ-state identity. Denoting this event by \(\mathsf{Forge}\) and the protocol acceptance event by \(\mathsf{Acc}\), the corresponding probability is bounded by
\begin{equation*}
\Pr\left[\mathsf{Forge}\wedge\mathsf{Acc}\right]
\leq
2^{-H_{\min}^{\varepsilon}(O_{1\sim N} \mid S_{1\sim N} E)_{\rho_{|\Omega}}} .
\end{equation*}
 Consequently, the probability that a dishonest prover passes the protocol while lacking the required state identity is bounded by

$$
\varepsilon^{(N)}
\leq
2^{-N t + v \sqrt{N}}.
$$
This bound holds against arbitrary coherent attacks and therefore does not rely on an i.i.d. assumption.\\

\noindent\textbf{Extension to linear Bell operators}: This analysis extends to the Bell operator $\mathcal{B}_n = -\hat{\mathcal{O}}_0 + \sum_{i=1}^n \hat{\mathcal{O}}_i$, with quantum ceiling $\beta_{\mathrm{Q}} = n+1$ and local-realistic bound $\beta_{\mathrm{LR}} = n-1$. We parameterize the violation by the normalized Bell score $\nu(\beta) = [\beta - (n-1)]/2 \in [0, 1]$. In an $N$-round sequence with observed mean $\widehat{\beta} = N^{-1}\sum_{j=1}^N \beta_j$, the protocol accepts if $\widehat{\beta} \ge \beta^* \equiv \beta_{\mathrm{obs}} - \delta$ for confidence margin $\delta > 0$. The single-round min-entropy against adversary $E$ is bounded by the guessing probability:$$r(O \mid E; \beta) \equiv -\log_2 p_{\mathrm{guess}}(O \mid E; \beta),$$ where $p_{\mathrm{guess}}(O \mid E; \beta) = \sup_{\rho, \mathcal{M}} \{ p_{\mathrm{guess}}(O \mid E) : \mathrm{Tr}(\rho \hat{\mathcal{B}}_n) \ge \beta \}$ is computed via the Navascu\'es--Pironio--Ac\'in (NPA) hierarchy \cite{Navascu}. Linearizing via an affine min-tradeoff function tangent at $\beta^*$, $f(\beta) = r(\beta^*) + r'(\beta^*)(\beta - \beta^*)$, the accumulated smooth min-entropy satisfies$$H_{\min}^{\varepsilon}(O_{1\sim N} \mid E T_{1\sim N})_{\rho\vert{}\Omega} \ge N f(\beta^*) - \mathcal{O}(\sqrt{N}),$$yielding the guessing probability bound$$p_{\mathrm{guess}}^{\varepsilon}(O_{1\sim N} \mid E T_{1\sim N}, \Omega) \le 2^{-N f(\beta^*) + \mathcal{O}(\sqrt{N})},$$and the asymptotic rate $$\lim_{N \to \infty} -\frac{1}{N} \log_2 p_{\mathrm{guess}}^{\varepsilon} = f_{\rm GHZ}(\beta_{\mathrm{obs}})$$ as $\delta \to 0$. \\

\noindent\textbf{Zero-knowledge:} Statistical zero-knowledge requires an efficient simulator reproducing the verifier's view within trace distance\be\label{fidelity1}\frac{1}{2} \left\Vert{} \rho_V^{\mathrm{real}} - \rho_V^{\mathrm{sim}} \right\Vert{}_1 \le \varepsilon_{\mathrm{ZK}}.\ee

Combined with soundness error $P_{\mathrm{accept}}^{\mathrm{dishonest}} \le \varepsilon$, this guarantees composable finite-size security against coherent attacks. Eq.~\eqref{fidelity1} directly underpins our primary theoretical guarantee, formalised in the following theorem.

\begin{thm}[Zero-knowledge criterion]\label{Guess_thm}
If the GHZ Protocol~\ref{pro_m2} achieves a maximal Bell violation within tolerance $\varepsilon_{ZK} \ge 0$, namely $\langle \mathcal{B}\rangle_\rho \ge (1+n) - \varepsilon_{ZK}$ for the Bell operator $\mathcal{B} = -\hat{\mathcal{O}}_0 + \sum_{i\in[n]}\hat{\mathcal{O}}_i$ [Eq.~\eqref{GHZ_Bell}] with odd $n$, then the verifiers' reduced view satisfies
\begin{equation*}\label{fidelity2}
\frac{1}{2}\left\| \rho_{V}^{\mathrm{real}} - \rho_{V}^{\mathrm{sim}} \right\|_1 \le \frac{\varepsilon_{\mathrm{ZK}}}{n},
\end{equation*}
where $\rho_{V}^{\mathrm{real(sim)}} = \mathrm{Tr}_{P}[\rho_{PV}^{\mathrm{real(sim)}}]$ with $V = V_2 V_3 \dots V_n$.
\end{thm}

\begin{proof}
 The trace distance admits the operational variational form
\begin{equation*}
D(\rho_V^{\rm{real}},\rho_V^{\rm{sim}}) = \frac{1}{2}\|\rho_V^{\rm{real}}-\rho_V^{\rm{sim}}\|_1 = \max_{0 \le \Pi \le \mathbb{I}} \mathrm{Tr}\big[\Pi(\rho_V^{\rm{real}}-\rho_V^{\rm{sim}})\big],
\end{equation*}
which quantifies the optimal single-shot distinguishing probability. Let $\Pi_V^\star$ denote an optimal measurement operator satisfying $0 \le \Pi_V^\star \le \mathbb{I}_V$, such that
\begin{equation*}
D(\rho_V^{\rm{real}},\rho_V^{\rm{sim}}) = \mathrm{Tr}_V\left[\Pi_V^\star(\rho_V^{\rm{real}}-\rho_V^{\rm{sim}})\right].
\end{equation*}
Exploiting the duality between the partial trace and the identity channel, $\mathrm{Tr}_V[\Pi_V \mathrm{Tr}_P(\rho_{PV})] = \mathrm{Tr}_{PV}[( \mathbb{I}_P\otimes \Pi_V)\rho_{PV}]$, we express this local distinguishability directly on the global space:
\begin{equation*}
D(\rho_V^{\rm{real}},\rho_V^{\rm{sim}}) = \mathrm{Tr}_{PV}\left[( \mathbb{I}_P\otimes \Pi_V^\star )(\rho_{PV}^{\rm{real}},\rho_{PV}^{\rm{sim}})\right].
\end{equation*}
Since $0 \le \Pi_V^\star \le \mathbb{I}_V$, the extended observable satisfies $0 \le \mathbb{I}_P \otimes \Pi_V^\star \le \mathbb{I}_{PV}$, thus constituting a valid positive operator-valued measure (POVM) element on $\mathcal{H}_P \otimes \mathcal{H}_V$. Since the global trace distance optimizes over all POVM elements on the composite Hilbert space, the restricted local strategy cannot outperform the global optimum:
\ben \begin{aligned}
D\left(\rho_V^{\rm{real}},\rho_V^{\rm{sim}}\right) &= \mathrm{Tr}_{PV}\left[\left( \mathbb{I}_P\otimes \Pi_V^\star \right)\left(\rho_{PV}^{\rm{real}}-\rho_{PV}^{\rm{sim}}\right)\right]\\
&\le \max_{0 \le \Pi_{PV} \le \mathbb{I}{PV}} \mathrm{Tr}_{PV}\left[\Pi_{PV}\left(\rho_{PV}^{\rm{real}}-\rho_{PV}^{\rm{sim}}\right)\right] \\
&= D\left(\rho_{PV}^{\rm{real}},\rho_{PV}^{\rm{sim}}\right),
\end{aligned} \een
establishing the monotonicity of trace distance under the partial trace. Applying the triangle inequality with respect to the target state $\vert{}\phi_n^-\rangle$ and by using Lemma \ref{Guess_lem} we have 
\ben
\begin{aligned}
 D\left(\rho_V^{\rm{real}},\rho_V^{\rm{sim}}\right)&\le D\left(\rho_{PV}^{\rm{real}},\rho_{PV}^{\rm{sim}}\right)\\
 &\le D\left(\rho_{PV}^{\rm{real}},\proj{\phi_n^-}\right)+ D\left(\proj{\phi_n^-},\rho_{PV}^{\rm{sim}}\right)\\
 &\le \frac{\varepsilon_{\rm{ZK}}}{2n}+\frac{\varepsilon_{\rm{ZK}}}{2n}=\frac{\varepsilon_{\rm{ZK}}}{n}.
\end{aligned}
\een
\end{proof}

\begin{lem}\label{Guess_lem}
If a state $\rho$ achieves maximal Bell violation with deficiency $\varepsilon \ge 0$, such that $\langle \mathcal{B}\rangle_\rho \ge (1+n)-\varepsilon$ for the Bell operator $\mathcal{B}$ defined in Eq.~\eqref{GHZ_Bell} with odd $n$, the trace distance to the ideal state satisfies
\begin{equation*}\label{fidelity}
\frac{1}{2}\left| \|\phi_n^-\rangle\langle\phi_n^-| - \rho \right\|_1 \le \frac{\varepsilon}{2n}.
\end{equation*}
\end{lem}

\begin{proof}
 By Lemma~\ref{lemma1}, any such state admits the decomposition
 \begin{equation*}\rho = (1 - \eta) |\phi_n^-\rangle \langle \phi_n^-| + \eta \rho_\perp,\quad \eta\ge 0.\end{equation*} The trace distance gives us\ben \begin{aligned}
 D\left( \proj{\phi_n^-}, \rho\right) &= \frac{1}{2}\big\Vert{} \vert{}\phi_n^-\rangle\langle\phi_n^-\vert{} - \rho \big\Vert{}_1 \\
 &= \frac{\eta}{2} \left\|\vert{}\phi_n^-\rangle\langle\phi_n^-\vert{} - \rho_\perp\right\|_1\\
 &=\frac{\eta}{2} \left\|\Delta\right\|_1\rm{ (say)},\end{aligned}\een
where $\Delta= \vert{}\phi_n^-\rangle\langle\phi_n^-\vert{} - \rho_\perp$. By definition, the orthogonal component satisfies $\rho_\perp \ge 0$, $\mathrm{Tr}(\rho_\perp) = 1$, and $\langle \phi_n^- \vert{} \rho_\perp \vert{} \phi_n^- \rangle = 0$, ensuring that $\mathrm{supp}(\vert{}\phi_n^-\rangle\langle\phi_n^-\vert{}) \perp \mathrm{supp}(\rho_\perp)$. The difference operator $\Delta = \vert{}\phi_n^-\rangle\langle\phi_n^-\vert{} - \rho_\perp$ therefore decouples into mutually orthogonal subspaces. Specifically, $\Delta$ possesses a single positive eigenvalue $\lambda = 1$ supported on $\mathrm{span}(\vert{}\phi_n^-\rangle)$, while on the orthogonal complement $\mathrm{span}(\vert{}\phi_n^-\rangle)^\perp$ it acts as $-\rho_\perp$ with nonpositive eigenvalues $\lambda_k = -e_k \le 0$, where $\{e_k\}$ denotes the spectrum of $\rho_\perp$ ($\sum_k e_k = 1$). Summing the absolute values of the eigenvalues directly yields the trace norm $\Vert{}\Delta\Vert{}_1 = 1 + \sum_k e_k = 2$. Hence, the trace distance simplifies to
\begin{equation}\label{eq:TD}
\frac{1}{2}\left\| |\phi_n^-\rangle\langle\phi_n^-| - \rho \right\|_1 = \eta.
\end{equation} Furthermore, evaluating the Bell expectation value for this state gives
\begin{equation*}
\langle\mathcal{B}\rangle_\rho = \mathrm{Tr}(\mathcal{B}\rho) \le (1+n) - 2n\eta.
\end{equation*}
Combining this upper bound with the hypothesis $\langle \mathcal{B}\rangle_\rho \ge (1+n) - \varepsilon$, we obtain
\(\eta \le \frac{\varepsilon}{2n}\). Substituting this constraint into Eq.~\eqref{eq:TD} completes the proof.
\end{proof}

\begin{thm}[Soundness criterion]
If the GHZ Protocol~\ref{pro_m2} achieves a maximal Bell violation with deficiency $\varepsilon \ge 0$, namely $\langle \mathcal{B}\rangle_\rho \ge (1+n) - \varepsilon$ for the Bell operator $\mathcal{B} = -\hat{\mathcal{O}}_0 + \sum_{i\in[n]}\hat{\mathcal{O}}_i$ [Eq.~\eqref{GHZ_Bell}] with odd $n$, then the soundness of the protocol in view of dishonest prover is bounded by 
\begin{equation*}P_{\mathrm{accept}}^{\mathrm{dishonest}} \le \varepsilon_{\mathrm{s}}=\exp\left[-N \ln\left(1 + \frac{2-\varepsilon}{n-1}\right)\right],
\end{equation*} for $N$ consecutive rounds.
\end{thm}

\begin{proof}

Under the local-realistic (LR) constraint, the Bell expectation value is bounded by $\beta_{\rm LR} \le n-1$, which limits the single-round acceptance probability for an arbitrary dishonest prover to $p_{\rm acc}^{(1),\mathrm{LR}} \le (n-1)/(n+1-\varepsilon)$. For $N$ rounds, the cumulative acceptance probability is strictly bounded by\begin{equation}\label{eq:soundness}\begin{aligned}P_{\rm accept}^{\rm dishonest} \le \varepsilon_{\rm s}(N,n)&\equiv \left(\frac{n-1}{n+1-\varepsilon}\right)^N \\&= \exp\left[-N \ln\left(1 + \frac{2-\varepsilon}{n-1}\right)\right].\end{aligned}\end{equation} Equation~\eqref{eq:soundness} establishes the soundness criterion of the protocol, demonstrating that the cheating probability is suppressed exponentially with the number of rounds $N$, with an asymptotic decay rate set by the quantum-classical margin $(2-\varepsilon)/(n-1)$.
\end{proof}
\noindent\textbf{Completeness}: For states orthogonal to the target state $\vert{}\phi_n^-\rangle$, the GHZ paradox \eqref{GHZ_P} guarantees that at least one of the $n+1$ stabilizer constraints is violated, giving $$p_{\rm acc}^{(1),\perp} \le \frac{n}{n+1}.$$ For a state $\rho = (1-\eta)\vert{}\phi_n^-\rangle\langle\phi_n^-\vert{} + \eta\rho_\perp$, the single-round acceptance probability for honest prover satisfies $$1-\eta\leq p_{\rm acc}^{(\rm honnest)}(\rho) \le 1 - \frac{\eta}{n+1},$$ showing that any non-target weight $\eta > 0$ induces a finite acceptance.\\
\noindent\textbf{Adversarial guessing bounds:} Security against dishonest provers demands an upper bound on the adversary's probability of correctly guessing the verifiers' outcomes, subject to either the GHZ paradox constraints [Eq.~\eqref{noise}] or the corresponding Bell inequality violation. In the tripartite setting, local realism and quantum theory constrain the parameter to $1/2 \ge \epsilon \ge 0$. Formulating the NPA hierarchy~\cite{Navascu} as a semidefinite program upper bounds the joint guessing probability for verifiers $V_2$ and $V_3$ (see Fig.~\ref{fig:Two_party_with_eps}). Under these constraints, the NPA hierarchy yields the analytical upper bound on the conditional guessing probability:  
\begin{equation*}
 p_{\mathrm{guess}}(V_2 V_3 \mid P_1) \le
\begin{cases}
\dfrac{1}{4} + \dfrac{\epsilon}{2} + \dfrac{\sqrt{3}}{2}\sqrt{\epsilon(1-\epsilon)}, & 0 \le \epsilon \le \dfrac{1}{4}, \\[3mm]
\dfrac{1}{2} + \epsilon, & \dfrac{1}{4} \le \epsilon \le \dfrac{1}{2}.
 \label{eq:f2bound}
\end{cases} 
\end{equation*}
\begin{figure}
 \centering
 \includegraphics[width=0.8\columnwidth]{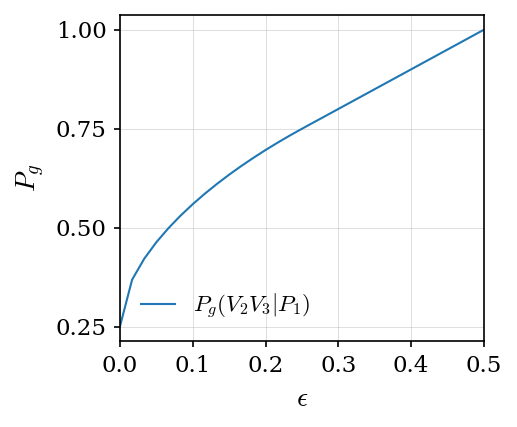}
 \caption{Guessing probability $P_g(V_2V_3|P_1)$ as a function of
 $\epsilon$, where $\epsilon$ quantifies the deviation from the GHZ paradox constraints [Eq.(\ref{noise})].}
 \label{fig:Two_party_with_eps}
\end{figure}

\section{Conclusion}
In this work, we have developed an information-theoretically secure framework for zero-knowledge certification of shared entangled states and constructed explicit bipartite and multipartite protocols based on Bell and GHZ resources. We have shown that a naive Bell-state certification protocol admits a correlation loophole, which motivates a two-basis certification scheme and, subsequently, a device-independent formulation capable of excluding higher-dimensional realizations. Extending this approach to quantum networks, we employ recently introduced self-testing of GHZ correlations and measurement structure, thereby establishing device-independent zero-knowledge protocols. The zero-knowledge property is proven by an explicit simulation of the verifier's view, without access to the prover's secret information, while completeness and soundness are maintained independently of computational assumptions. We further demonstrate that these guarantees persist under noise, establishing robustness of the protocols against imperfections in the shared quantum resources. Our results connect entanglement certification, nonlocality, self-testing, and zero-knowledge verification, and provide a foundation for cryptographic primitives in distributed quantum networks.
\section{Acknowledgment}
R. Rahaman acknowledges support from the ANRF Advanced Research Grant (ARG), Grant No. ANRF/ARG/2025/012066/MS. 
\bibliography{ref}

@article{GMR89,
  author    = {Goldwasser, Shafi and Micali, Silvio
and  Rackoff, Charles},
  title     = {The Knowledge Complexity of Interactive Proof Systems},
  journal   = {SIAM Journal on Computing},
  pages     = {186–208},
  year      = {1989},
  doi = {10.1137/0218012},
  url = {https://doi.org/10.1137/0218012}
}

@inproceedings{BenSasson2014,
  author    = {Eli Ben-Sasson and Alessandro Chiesa and Christina Garman and
               Matthew Green and Ian Miers and Eran Tromer and Madars Virza},
  title     = {Zerocash: Decentralized Anonymous Payments from Bitcoin},
  booktitle = {2014 IEEE Symposium on Security and Privacy},
  pages     = {459--474},
  year      = {2014},
  publisher = {IEEE},
  doi        = {10.1109/SP.2014.36}
}

@inproceedings{Bulletproofs2018,
  author    = {Benedikt B{\"u}nz and Jonathan Bootle and Dan Boneh and
               Andrew Poelstra and Pieter Wuille and Greg Maxwell},
  title     = {Bulletproofs: Short Proofs for Confidential Transactions and More},
  booktitle = {2018 IEEE Symposium on Security and Privacy},
  pages     = {315--334},
  year      = {2018},
  publisher = {IEEE},
  doi       = {10.1109/SP.2018.00020}
}

@book{Narayanan2016,
  author    = {Arvind Narayanan and Joseph Bonneau and Edward Felten and
               Andrew Miller and Steven Goldfeder},
  title     = {Bitcoin and Cryptocurrency Technologies: A Comprehensive Introduction},
  publisher = {Princeton University Press},
  address   = {Princeton, NJ},
  year      = {2016},
  isbn      = {9780691171692}
}

@artcile{Blckchn1,
  author={Sun, Xiaoqiang and Yu, F. Richard and Zhang, Peng and Sun, Zhiwei and Xie, Weixin and Peng, Xiang},
  journal={IEEE Network}, 
  title={A Survey on Zero-Knowledge Proof in Blockchain}, 
  year={2021},
  volume={35},
  number={4},
  pages={198-205},
  doi={10.1109/MNET.011.2000473}}

@article{blckchn2,
  author={Partala, Juha and Nguyen, Tri Hong and Pirttikangas, Susanna},
  journal={IEEE Access}, 
  title={Non-Interactive Zero-Knowledge for Blockchain: A Survey}, 
  year={2020},
  volume={8},
  number={},
  pages={227945-227961},
  doi={10.1109/ACCESS.2020.3046025}}

@article{RSA77,
author = {Rivest, R. L. and Shamir, A. and Adleman, L.},
title = {A method for obtaining digital signatures and public-key cryptosystems},
year = {1978},
issue_date = {Feb. 1978},
publisher = {Association for Computing Machinery},
address = {New York, NY, USA},
volume = {21},
number = {2},
issn = {0001-0782},
url = {https://doi.org/10.1145/359340.359342},
doi = {10.1145/359340.359342},
journal = {Commun. ACM},
month = feb,
pages = {120–126},
numpages = {7},

}

@article{DH76,
  author={Diffie, W. and Hellman, M.},
  journal={IEEE Transactions on Information Theory}, 
  title={New directions in cryptography}, 
  year={1976},
  volume={22},
  number={6},
  pages={644-654},
  doi={10.1109/TIT.1976.1055638}}

@article{SHOR97,
author = {Shor, Peter W.},
title = {Polynomial-Time Algorithms for Prime Factorization and Discrete Logarithms on a Quantum Computer},
journal = {SIAM Journal on Computing},
volume = {26},
number = {5},
pages = {1484-1509},
year = {1997},
doi = {10.1137/S0097539795293172},
URL ={https://doi.org/10.1137/S0097539795293172},
}

@article{Grover97,
author = {Grover, Lov K.},
title = {A fast quantum mechanical algorithm for database search},
year = {1996},
isbn = {0897917855},
journal = {Proceedings of the 28th Annual ACM Symposium on Theory of Computing},
publisher = {Association for Computing Machinery},
address = {New York, NY, USA},
url = {https://doi.org/10.1145/237814.237866},
doi = {10.1145/237814.237866},
booktitle = {Proceedings of the Twenty-Eighth Annual ACM Symposium on Theory of Computing},
pages = {212–219},
numpages = {8},
location = {Philadelphia, Pennsylvania, USA},
series = {STOC '96}
}

@incollection{JOHNSON90,
  author    = {Johnson, David S.},
  title     = {A Catalog of Complexity Classes},
  booktitle = {Handbook of Theoretical Computer Science, Volume A: Algorithms and Complexity},
  editor    = {van Leeuwen, Jan},
  publisher = {Elsevier},
  year      = {1990},
  pages     = {67--161},
  chapter   = {2},
  doi       = {10.1016/B978-0-444-88071-0.50007-2}
}

@book{AroraBarak2009,
  author    = {Sanjeev Arora and Boaz Barak},
  title     = {Computational Complexity: A Modern Approach},
  publisher = {Cambridge University Press},
  address   = {Cambridge},
  year      = {2009},
  isbn      = {9780521424264}
}

@book{GareyJohnson1979,
  author    = {Michael R. Garey and David S. Johnson},
  title     = {Computers and Intractability: A Guide to the Theory of NP-Completeness},
  publisher = {W. H. Freeman},
  address   = {San Francisco},
  year      = {1979},
  isbn      = {9780716710455}
}

@book{Goldreich2008,
  author    = {Oded Goldreich},
  title     = {Computational Complexity: A Conceptual Perspective},
  publisher = {Cambridge University Press},
  address   = {Cambridge},
  year      = {2008},
  isbn      = {9780521884730}
}

@article{walgate2000,
  title = {Local Distinguishability of Multipartite Orthogonal Quantum States},
  author = {Walgate, Jonathan and Short, Anthony J. and Hardy, Lucien and Vedral, Vlatko},
  journal = {Phys. Rev. Lett.},
  volume = {85},
  issue = {23},
  pages = {4972--4975},
  numpages = {0},
  year = {2000},
  month = {Dec},
  publisher = {American Physical Society},
  doi = {10.1103/PhysRevLett.85.4972},
  url = {https://link.aps.org/doi/10.1103/PhysRevLett.85.4972}
}

@article{kar01,
  title = {Distinguishability of Bell States},
  author = {Ghosh, Sibasish and Kar, Guruprasad and Roy, Anirban and Sen(De), Aditi and Sen, Ujjwal},
  journal = {Phys. Rev. Lett.},
  volume = {87},
  issue = {27},
  pages = {277902},
  numpages = {2},
  year = {2001},
  month = {Dec},
  publisher = {American Physical Society},
  doi = {10.1103/PhysRevLett.87.277902},
  url = {https://link.aps.org/doi/10.1103/PhysRevLett.87.277902}
}

@article{walgate02,
  title = {Nonlocality, Asymmetry, and Distinguishing Bipartite States},
  author = {Walgate, Jonathan and Hardy, Lucien},
  journal = {Phys. Rev. Lett.},
  volume = {89},
  issue = {14},
  pages = {147901},
  numpages = {4},
  year = {2002},
  month = {Sep},
  publisher = {American Physical Society},
  doi = {10.1103/PhysRevLett.89.147901},
  url = {https://link.aps.org/doi/10.1103/PhysRevLett.89.147901}
}

@article{Kimble2008,
  author    = {H. Jeff Kimble},
  title     = {The Quantum Internet},
  journal   = {Nature},
  volume    = {453},
  pages     = {1023--1030},
  year      = {2008},
  doi       = {10.1038/nature07127}
}

@article{Wehner2018,
  author    = {Stephanie Wehner and David Elkouss and Ronald Hanson},
  title     = {Quantum Internet: A Vision for the Road Ahead},
  journal   = {Science},
  volume    = {362},
  number    = {6412},
  pages     = {eaam9288},
  year      = {2018},
  doi       = {10.1126/science.aam9288}
}

@article{Bell64,
  title = {On the Einstein Podolsky Rosen paradox},
  author = {Bell, J. S.},
  journal = {Physics Physique Fizika},
  volume = {1},
  issue = {3},
  pages = {195--200},
  numpages = {6},
  year = {1964},
  month = {Nov},
  publisher = {American Physical Society},
  doi = {10.1103/PhysicsPhysiqueFizika.1.195},
  url = {https://link.aps.org/doi/10.1103/PhysicsPhysiqueFizika.1.195}
}

@article{CHSH,
  title = {Proposed Experiment to Test Local Hidden-Variable Theories},
  author = {Clauser, John F. and Horne, Michael A. and Shimony, Abner and Holt, Richard A.},
  journal = {Phys. Rev. Lett.},
  volume = {23},
  issue = {15},
  pages = {880--884},
  numpages = {0},
  year = {1969},
  month = {Oct},
  publisher = {American Physical Society},
  doi = {10.1103/PhysRevLett.23.880},
  url = {https://link.aps.org/doi/10.1103/PhysRevLett.23.880}
}

@incollection{GHZ,
  author    = {Greenberger, Daniel M. and Horne, Michael A. and Zeilinger, Anton},
  title     = {Going Beyond Bell's Theorem},
  booktitle = {Bell's Theorem, Quantum Theory and Conceptions of the Universe},
  editor    = {Kafatos, Menas},
  publisher = {Springer Netherlands},
  address   = {Dordrecht},
  pages     = {69--72},
  year      = {1989},
  isbn      = {978-94-017-0849-4},
  doi       = {10.1007/978-94-017-0849-4_10},
  url       = {https://doi.org/10.1007/978-94-017-0849-4_10}
}

@article{Watrous08,
  author    = {Watrous, John},
  title     = {Zero-knowledge against quantum attacks},
  journal   = {Proceedings of the Thirty-Eighth Annual ACM Symposium on Theory of Computing},
  pages     = {296–305},
  year      = {2006},
  doi = {10.1145/1132516.1132560},
  url = {https://doi.org/10.1145/1132516.1132560}
}

@inproceedings{Goldwasser1985Knowledge,
  author    = {Shafi Goldwasser and Silvio Micali and Charles Rackoff},
  title     = {The Knowledge Complexity of Interactive Proof-Systems},
  booktitle = {Proceedings of the Seventeenth Annual ACM Symposium on Theory of Computing (STOC '85)},
  pages     = {291--304},
  year      = {1985},
  publisher = {Association for Computing Machinery},
  address   = {Providence, Rhode Island, USA},
  doi       = {10.1145/22145.22178}
}

@article{Das2025, author = {Srijani Das and Manasi Patra and Tuhin Paul and Anish Majumdar and Ramij Rahaman}, title = {Device-Independent Anonymous Communication in Quantum Networks}, journal = {arXiv preprint}, year = {2025}, doi = {10.48550/arXiv.2512.21047} }

@misc{Yao,
      title={Self testing quantum apparatus}, 
      author={Dominic Mayers and Andrew Yao},
      year={2004},
      eprint={quant-ph/0307205},
      archivePrefix={arXiv},
      primaryClass={quant-ph},
      url={https://arxiv.org/abs/quant-ph/0307205}, 
}

@article{masaness,
  title = {Asymptotic Violation of Bell Inequalities and Distillability},
  author = {Masanes, Llu\'{\i}s},
  journal = {Phys. Rev. Lett.},
  volume = {97},
  issue = {5},
  pages = {050503},
  numpages = {4},
  year = {2006},
  month = {Aug},
  publisher = {American Physical Society},
  doi = {10.1103/PhysRevLett.97.050503},
}

@article{Arnon-Friedman2018,
  author  = {Arnon-Friedman, Rotem and Dupuis, Fr{\'e}d{\'e}ric and Fawzi, Omar and Renner, Renato and Vidick, Thomas},
  title   = {Practical device-independent quantum cryptography via entropy accumulation},
  journal = {Nature Communications},
  volume  = {9},
  number  = {1},
   pages   = {459},
  year    = {2018},
  doi     = {10.1038/s41467-017-02307-4},
  
}

@article{Navascu,
  title = {Bounding the Set of Quantum Correlations},
  author = {Navascu\'es, Miguel and Pironio, Stefano and Ac\'{\i}n, Antonio},
  journal = {Phys. Rev. Lett.},
  volume = {98},
  issue = {1},
  pages = {010401},
  numpages = {4},
  year = {2007},
  month = {Jan},
  publisher = {American Physical Society},
  doi = {10.1103/PhysRevLett.98.010401},
  url = {https://link.aps.org/doi/10.1103/PhysRevLett.98.010401}
}

@article{Tomamichel2009AEP,
  author  = {Tomamichel, Marco and Colbeck, Roger and Renner, Renato},
  title   = {A Fully Quantum Asymptotic Equipartition Property},
  journal = {IEEE Transactions on Information Theory},
  volume  = {55},
  number  = {12},
  pages   = {5840--5847},
  year    = {2009},
  month   = dec,
  doi     = {10.1109/TIT.2009.2032797},
  eprint  = {0811.1221},
  archivePrefix = {arXiv},
  primaryClass  = {quant-ph}
}

\end{document}